\documentclass{elsarticle}

\usepackage[american]{babel}	%English hyphenation etc.
\usepackage[utf8]{inputenc}		%Input encoding
\usepackage[T1]{fontenc}		%Font encoding
\usepackage[english=american]{csquotes}	%Quotation marks ' '
\usepackage{graphicx}			%Include graphics
\usepackage{epstopdf}
\usepackage{subfig}				%Subfigures
\usepackage{marvosym}			%Special symbols
\usepackage{amsmath}			%Mathmode
\usepackage{amssymb}
\usepackage{amsthm}
\usepackage{mathtools}
\usepackage{color}

\usepackage{mathpazo}

\usepackage{multirow}

\usepackage[splitrule]{footmisc}
\DeclareMathOperator*{\argmin}{argmin}

\newtheorem{conjecture}{Conjecture}
\newtheorem*{problem}{Problem}
\newtheorem{lemma}{Lemma}
\newtheorem{theorem}{Theorem}
\newtheorem{proposition}{Proposition}
\newtheorem{corollary}{Corollary}

\begin{document}

\begin{frontmatter}
\title{On the uniqueness of the maximum parsimony tree for data with up to two substitutions: an extension of the classic Buneman theorem in phylogenetics}

\author{Mareike Fischer\corref{cor1}}
\ead{email@mareikefischer.de}
\cortext[cor1]{Corresponding author}

\address{Institute of Mathematics and Computer Science, Greifswald University, Greifswald, Germany}

\begin{abstract} One of the main aims of phylogenetics is the reconstruction of the correct evolutionary tree when data concerning the underlying species set are given. These data typically come in the form of DNA, RNA or protein alignments, which consist of various characters (also often referred to as sites). Often, however, tree reconstruction methods based on criteria like maximum parsimony may fail to provide a unique tree for a given dataset, or, even worse, reconstruct the `wrong' tree (i.e. a tree that differs from the one that generated the data). On the other hand it has long been known that if the alignment consists of all the characters that correspond to edges of a particular tree, i.e. they all require exactly $k=1$ substitution to be realized on that tree, then this tree will be recovered by maximum parsimony methods. This is based on Buneman's theorem in mathematical phylogenetics. 
It is the goal of the present manuscript to extend this classic result as follows: We prove that if an alignment consists of all characters that require exactly $k=2$ substitutions on a particular tree, this tree will always be the unique maximum parsimony tree (and we also show that this can be generalized to characters which require at most $k=2$ substitutions). In particular, this also proves a conjecture based on a recently published observation by Goloboff et al. affirmatively for the special case of $k=2$.

\end{abstract}

\end{frontmatter}

%\linenumbers

\section{Introduction}

Mathematical phylogenetics is concerned with reconstructing the evolutionary relationships of a species set $X$ based on data. Traditionally, these relationships are represented by a phylogenetic tree and the data comes in the form of an alignment (e.g. aligned DNA, RNA or proteins or aligned binary sequences like absence or presence of certain morphological characteristics), whose columns are also often referred to as characters or sites. While no tree reconstruction method can guarantee to recover the true tree for all data sets, it has long been known that in some special cases a tree can be uniquely recovered. One such example is due to the classic theorem by Buneman \citep{Buneman1971}. 

The Buneman theorem at first glance has nothing to do with data. It states that any list of compatible $X$-splits corresponds to precisely one phylogenetic tree $T$. Here, an $X$-split, which is a bipartition of the species set, can be regarded as an edge of the tree (because each edge splits the species set into two disjoint and non-empty subsets). Now if you encode these $X$-splits as binary characters (where species in the same subset are assigned the same state) and summarize them in an alignment which we call $A_1(T)$, it is immediately clear that the Buneman theorem states that these binary data now correspond to a unique tree, namely $T$. 

Moreover, it has long been known that this unique tree can be recovered even by simple methods like those  based on the maximum parsimony principle. This is due to the fact that these characters are all compatible and therefore the unique tree which exists due to Buneman is a perfect phylogeny for the data, i.e. a tree which is compatible with all characters under consideration \citep[p. 69]{Semple2003}. This result still holds if constant characters are added to the data, i.e. characters which assign the same state to all species. 

The above mentioned maximum parsimony principle seeks the tree which requires as few character state changes along its edges as possible, i.e. in this sense it tries to minimize the number of mutations/substitutions needed to explain the evolution of the species set under investigation. For the particular alignment as constructed above, which consists precisely of the characters induced by the edges of a particular tree $T$ plus possibly some constant characters, this means that maximum parsimony would find the correct tree as, by the Buneman theorem, $T$ is the only tree that can represent all of these binary characters with precisely one change (which is best possible for a binary character) and the constant characters with 0 changes. If we denote the set of constant characters by $A_0$ and the concatenation of $A_0$ and $A_1(T)$ with $A_0 . A_1(T)$, then the unique maximum parsimony tree for $A_0 . A_1(T)$ (as well as for $A_1(T)$) is $T$. 

Mathematically, it is a natural question if this result can be generalized to $A_k(T)$, where $k$ denotes the number of changes the characters in this alignment require on $T$. It has been recently conjectured that this is indeed possible as long as $k<\frac{n}{4}$, where $n$ denotes the number of species under investigation \citep{pablo}. Biologically, this question is of interest because maximum parsimony is often assumed to be justified when the number of evolutionary events like substitutions are rare (c.f. for instance \citep[Chapter 5]{Semple2003}). So if we consider all characters that have precisely $k$ or at most $k$ changes on a given binary phylogenetic tree $T$, can maximum parsimony then recover the tree from this data set? Answering this question could also shed more light on conditions like e.g. the distribution of homoplasy required for maximum parsimony trees to coincide with maximum likelihood trees \citep{pablo}.

In this manuscript, we focus on the special case of $k=2$ and answer the question affirmatively. We do this by first extending the Buneman theorem (or, more precisely, the above described aspect of the Buneman theorem which deals with the recoverability of the correct tree) from binary characters with one change to binary characters with two changes. We also show that the Buneman theorem cannot be further extended in the same way to more than two changes, which is why a proof of the conjecture for $k\geq 3$ will require a different approach.

\section{Preliminaries}
\subsection{Notation}

We start with some notation. Recall that a \emph{phylogenetic tree} $T=(V,E)$ on a species or taxon set $X$ is a connected acyclic graph with vertex set $V$ and edge set $E$ whose leaves are bijectively labelled by $X$. We may assume without loss of generality that $X=\{1, \ldots, n\}$. A phylogenetic tree $T$ is called \emph{binary} if all inner nodes have degree 3. Throughout this manuscript, unless stated otherwise, when we refer to a tree $T$, we always mean a binary phylogenetic $X$-tree. However, while the trees we are interested in are unrooted, for technical reasons we sometimes also have to consider rooted trees: When an edge is removed from a binary (unrooted) tree, two subtrees remain, both of which have precisely one node of degree 2. This node is considered the {\em root} of the respective subtree. In a rooted binary phylogenetic tree, the two trees that you obtain when you delete the root node and both edges adjacent to the root are called {\em maximal pending subtrees} of this rooted tree. Figure \ref{treedecomp} illustrates these notions.

\begin{figure} 
\center
%\scalebox{.2}{ \input{treedecomp.pstex_t} }
\scalebox{.17}{\includegraphics{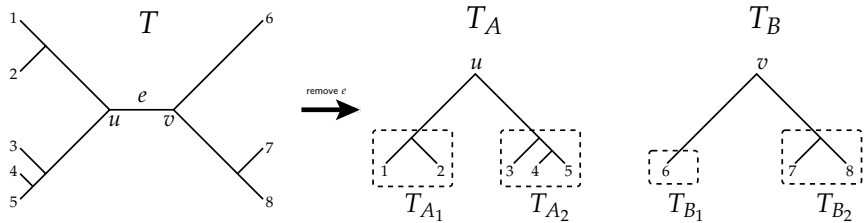}}
\caption{  \scriptsize By removing an edge $e$ from an unrooted phylogenetic tree $T$, it is decomposed into two rooted subtrees, $T_A$ and $T_B$. If, as in this figure, both of them consist of more than one node, then we can further decompose them into their two maximal pending subtrees, $T_{A_1}$ and $T_{A_2}$ or $T_{B_1}$ and $T_{B_2}$, respectively. }
\label{treedecomp}
\end{figure}

In the present manuscript, we also need the concept of distances between leaves. We say that two leaves $x$ and $y$ are at distance $d$ in a binary phylogenetic tree $T$, i.e. $d_T(x,y)=d$, whenever the unique path from $x$ to $y$ in $T$ consists of $d$ edges. Moreover, we say that two leaves $v$ and $w$ form a \emph{cherry} $[v,w]$, if $v$ and $w$ are adjacent to the same inner node $u$ of $T$. In this case, $u$ is also called the \emph{parent} of $v$ and $w$. Note that for a cherry $[v,w]$, we have $d_T(v,w)=2$, and that, on the other hand, if $v$ and $w$ do not form a cherry, we necessarily have $d_T(v,w)>2$. 

Furthermore, recall that a bipartition $\sigma$ of $X$ into two non-empty disjoint subsets $A$ and $B$ is often called \emph{$X$-split}, and is denoted by $\sigma=A|B$. Recall that there is a  natural relationship between $X$-splits and the edges of a phylogenetic $X$-tree $T$, because the removal of an edge $e$ induces a bipartition of $X$. In the following, the set of all such induced $X$-splits of $T$ will be denoted by $\Sigma(T)$. Recall that for a binary phylogenetic $X$-tree $T$ with $|X|=n$ we have $|\Sigma(T)|=2n-3$ \citep[Prop. 2.1.3]{Semple2003}. Moreover, note that the {\em size of an $X$-split $\sigma=A|B$} is defined as $|\sigma|=\min\{|A|,|B|\}$ \citep{Fischer2015a}. An $X$-split of size 1 is called {\em trivial}. Given a set of $X$-splits, an element of this set with minimal size is called a {\em minimal split}. 

Now that we have introduced the concept of a tree, we need to introduce the data. The data comes in the form of \emph{characters}, where a character $f$ is a function from the taxon set $X$ to a set $\mathcal{C}$ of character states, i.e. $f: X \rightarrow \mathcal{C}$. Note that a finite sequence of characters is also often referred to as \emph{alignment} in biology. In this case, the characters form the columns of an alignment and are also often called `sites'.
In this manuscript, we will only be concerned with \emph{binary characters}, i.e. without loss of generality $\mathcal{C}= \{a,b\}$. Instead of writing $f(1)=a$, $f(2)=a$, $f(3)=b$ and $f(4)=b$, we use the short form $f=aabb$. There is a close relationship between $X$-splits and binary characters, because every $X$-split can be represented by a binary character by assigning the same state to taxa in the same subset. For instance, if $\sigma= 12 | 34$, then characters $f_1=aabb$ and $f_2=bbaa$ would correspond to $\sigma$. If an $X$-split $\sigma_e$ is induced by an edge $e$ of a phylogenetic $X$-tree in the manner explained above, we also say that the corresponding binary character is induced by $e$. If an $X$-split is trivial, it must correspond to an edge that leads to a leaf of any tree, because it only separates one taxon from the other taxa. 

It is important to note that in this manuscript, whenever two binary characters refer to the same $X$-split, we regard them as identical. This means that we do not distinguish between $f=aabb$ and $f=bbaa$, for instance. Therefore, throughout this manuscript, we assume for technical reasons and without loss of generality that $f(1)=a$. 

\par\vspace{0.5cm}
So a character $f: X \rightarrow \{a,b\}$ assigns states to all leaves of the tree. If the inner nodes of a tree are also to be assigned states, we need an \emph{extension} of the character. An extension of a binary character $f$ on a phylogenetic tree $T$ with vertex set $V$ is a map $g: V \rightarrow \{a,b\}$ such that $g(x)=f(x)$ for all $x \in X$. Moreover, we call $ch(g) = \vert \{ (u,v) \in E, \, g(u) \neq g(v)\} \vert$ the \emph{changing number} of $g$ on $T$.

\par\vspace{0.5cm}
Now that we have established the kind of data we consider, namely characters, as well as the object we want to reconstruct, namely phylogenetic trees, we need to introduce a method to do just that. This manuscript uses the \emph{maximum parsimony} principle to infer trees from characters: Given a character $f$, the idea of maximum parsimony is to find a phylogenetic tree $T$ that minimizes the so-called \emph{parsimony score} $l(f,T)$ of $f$, where $l(f,T) = \min\limits_{g} ch(g,T)$ and where the minimum runs over all extensions $g$ of $f$ on $T$. The parsimony score of an alignment $A=\{f_1,\ldots, f_m\}$ is then defined as: $l(A,T)=\sum\limits_{i=1}^m l(f_i,T)$. Moreover, a \emph{maximum parsimony tree} $T$ of an alignment $A$ is defined as $T=\argmin_{\tilde{T}} l(A,\tilde{T})$. 

For a given tree $T$ and a character $f$, an extension $g$ that minimizes the parsimony score of $f$ on $T$ is called a \emph{most parsimonious extension} or sometimes also a  \emph{minimal extension}. There are several well-known algorithms to calculate the parsimony score for a given phylogenetic tree and a given character. For instance, the well-known Fitch algorithm can be used \citep{Fitch}. This algorithm works in polynomial time, i.e. finding the parsimony score of a character on a tree (which is often referred to as the `small parsimony problem') is easy.

However, the so-called `big parsimony problem', namely finding a maximum parsimony tree for an alignment, is known to be NP-complete \citep{foulds_graham_1982}. Moreover, there may be more than one maximum parsimony tree for an alignment, i.e. the maximum parsimony tree need not always be unique. 

\par\vspace{0.5cm}

In this manuscript, given a phylogenetic tree $T$, we are concerned with finding maximum parsimony trees for the alignment $A_k(T)$, which we define to be the set consisting of all binary characters that have parsimony score $k$ on $T$. In particular, we will consider the case $k=2$ and show that $T$ is the unique maximum parsimony tree of $A_2(T)$. Note that $A_1(T)$ corresponds to all characters induced by the edges of $T$, and that $A_0(T)$ consists precisely of $f=aa\ldots a$, i.e. the constant character, as this is the only character that has parsimony score 0. Moreover, as $A_0(T)=A_0(\tilde{T})$ for all $\tilde{T}$ on the same taxon set $X$, we usually write $A_0$ instead of $A_0(T)$. Note that when we concatenate character disjoint alignments like $A_1(T)$ and $A_0$, we denote this concatenation, i.e. the union of the character sets, by a dot, e.g. $A_0 . A_1(T)$. As we regard alignments merely as sets of characters, the order of the characters in the set does not matter. This is different in many more specialized biological models, where the exact position of a character in an alignment might have an impact on tree reconstruction, e.g. when different rates across sites are considered (c.f. for instance \citep{Susko2003}). However, as maximum parsimony makes no such assumptions, we do not need this restriction in the present manuscript.

\subsection{Known results}

A basic result that we need throughout this manuscript is the following theorem, which counts the number of characters in $A_k(T)$.

\begin{theorem} \label{thm:lengthAk} \citep{Steel1993paper,Book_Steel}
Let $T$ be a binary phylogenetic $X$-tree with $|X|=n$. Then, we have:
$$|A_k(T)|= \frac{1}{2} \cdot \frac{2n-3k}{k}{n-k-1\choose k-1} \cdot 2^k=\frac{2n-3k}{k}{n-k-1\choose k-1} \cdot 2^{k-1}.$$
\end{theorem}

Note that in the original version (cf. \citep[p. 101, eq. (5.7)]{Book_Steel}), the formula does not contain the factor $\frac{1}{2}$, which is due to the fact that the authors there count {\em all} binary characters with score $k$ on $T$, whereas $A_k(T)$ by our definition only contains those for which $f(1)=a$. 

Now, the most important case for this manuscript is $k=2$, in which case Theorem \ref{thm:lengthAk} gives $|A_2(T)|=2(n-3)^2$, where $n$ denotes the number of leaves of $T$. We will need this formula later on.

The most important theorem on which this manuscript is based is the following classic theorem by Buneman  \citep{Buneman1971} (see also \citep[p. 44]{Semple2003}).

\begin{theorem}[Buneman]\label{thm:buneman}
Let $T$ and $\tilde{T}$ be two binary phylogenetic $X$-trees. 
Then, $T=\tilde{T}$ if and only if $\Sigma(T)=\Sigma(\tilde{T})$. 
\end{theorem}

In particular, we will consider the following corollary, which is a direct consequence of Theorem \ref{thm:buneman}.

\begin{corollary}\label{cor:buneman}
Let $T$ be a binary phylogenetic $X$-tree. Then, $T$ is the unique maximum parsimony tree for the alignment $A_1(T)$.
\end{corollary}

The proof of this corollary exploits the following lemma.

\begin{lemma} \label{lem:onecut} Let $T$ be a phylogenetic $X$-tree with $|X|=n$ and $f$ a binary character on $X$ such that $l(f,T)=1$. Then, there is an edge $e$ of $T$ such that the $X$-split $X=X_1 | X_2$ induced by $e$ is such that for all $x \in X_1$ and $y \in X_2$ we have $f(x)=a$ and $f(y)=b$ (or vice versa). Moreover, $|A_1(T)|=|\Sigma(T)|=2n-3$.
\end{lemma}

\begin{proof} As $l(f,T)=1$, there is a most parsimonious extension $g$ of $f$ on $T$ such that there is precisely one edge $e=\{u,v\}$ in $T$ for which $g(u)\neq g(v)$. Removing this edge partitions $T$ into two subtrees, one of which contains only nodes assigned $a$ by $g$ and the other one $b$. As $g(x)=f(x)$ for all leaves $x \in X$, this completes the first part of the proof. As this applies to all $f \in A_1(T)$, this immediately leads to $|A_1(T)|=|\Sigma(T)|$. This completes the proof.
\end{proof}

Now we are in the position to prove Corollary \ref{cor:buneman}. 

\begin{proof}[Proof of Corollary \ref{cor:buneman}] Let $T$ be a binary phylogenetic $X$-tree and let $E$ denote its edge set. Consider $A_1(T)$ and assume there is a binary phylogenetic $X$-tree $\tilde{T}$ such that $l(A_1(T),\tilde{T}) \leq l(A_1(T),T)=\sum\limits_{e \in E}1 = |E|$. For all $f \in A_1(T)$, we know that $f$ is not constant, because otherwise it would not have parsimony score 1 on $T$. So $f$ employs both character states $a$ and $b$. Therefore, it requires at least one change on {\em any} phylogenetic $X$-tree -- in particular also on $\tilde{T}$. Thus, $ l(A_1(T),\tilde{T})\geq\sum\limits_{f \in A_1(T)}1 =l(A_1(T),\tilde{T})= |E|$. So in summary, we must have  $l(A_1(T),\tilde{T}) = l(A_1(T),T)$. However, this implies in particular for all $f \in A_1(T)$ that $l(f,\tilde{T})=1$. So by Lemma \ref{lem:onecut}, all $f$ in $A_1(T)$ correspond to an edge of $\tilde{T}$. However, as $|A_1(T)|=|\Sigma(T)|=|\Sigma(\tilde{T})|$ (by Lemma \ref{lem:onecut} and using the fact that both $T$ and $\tilde{T}$ are binary), if all $f$ in $A_1(T)$, which by Lemma  \ref{lem:onecut} correspond to edges in $T$, also correspond to edges in $\tilde{T}$, then we have $\Sigma(T)=\Sigma(\tilde{T})$. This, by Theorem \ref{thm:buneman}, implies that $T=\tilde{T}$ and thus completes the proof.
\end{proof}

The last prerequisite that we need for this manuscript is the work presented in \citep{pablo}. In this manuscript, the authors analyzed all trees with up to $n=20$ leaves for $k=2$ and binary data, which corresponds to $A_k$ (and up to $n=12$ for non-binary data, which we are not considering here). They found that if $T$ is a binary phylogenetic $X$-tree with $|X|=n$ leaves and if $k< \frac{n}{4}$, then $T$ is the unique maximum parsimony tree for $A_k(T)$. This observation motivates the following conjecture.

\begin{conjecture} \label{conj} Let $T$ be a binary phylogenetic $X$-tree with $|X|=n$. Let $k< \frac{n}{4}$. Then, $T$ is the unique maximum parsimony tree for $A_k(T)$. 
\end{conjecture}

Note that the exhaustive search performed in \citep{pablo} shows in particular that the conjecture holds for $n=9$ and $k=2$, which we will use as the base case for our inductive proof.

In the present manuscript, we will mainly focus on the case $k=2$.\footnote{In \citep[p. 96]{pablo}, the authors state: ``For $t\geq 9$ and $s=2$, it can be verified that $T=P$ for every treeshape''. In their notation, $t$ is the number of leaves (i.e. $n$ in our case), $s$ is the number of changes (i.e. $k$ in our case), $T$ is the tree defining alignment $A_2(T)$ and $P$ is the -- in this case unique -- maximum parsimony tree of $A_2(T)$. However, note that this statement of the authors is solely based on their exhaustive search, and no formal proof is given in their manuscript. It is the main aim of the present manuscript to formally prove that this observation of \citep{pablo} is correct.} Regarding Conjecture \ref{conj}, this implies that we have to consider trees with more than eight leaves, as we require $2 < \frac{n}{4}$. Note that the motivation behind the requirement $k< \frac{n}{4}$ in Conjecture \ref{conj} is due to the fact that for larger values of $k$, it is already known that the conjecture fails. For instance, for $k=2$ and $n=8$, consider trees $T_1$ and $T_2$ as depicted in Figure \ref{8taxacounterex}. Note that the two trees differ only by swapped positions of the cherries $[3,4]$ and $[7,8]$. As stated in \citep{pablo} and as can be easily verified, for these two trees, Conjecture \ref{conj} does not hold. More specifically, as $A_2(T_1)$ consists of 50 characters, all of which have a parsimony score of 2 on $T_1$, we have $l(A_2(T_1),T_1)=100$. But, surprisingly, it can also be shown that we have $l(A_2(T_1),T_2)=99$. In Figure \ref{betterworse} we give an overview of the characters of $A_2(T_1)$. In fact, 49 of the 50 characters of $A_2(T_1)$ also have parsimony score 2 on $T_2$ and are thus also contained in $A_2(T_2)$, so they perform equally on both trees. Only one character, namely $f=aabbbbaa$, i.e. the one that refers to the center edge of $T_2$, has a parsimony score of only 1 on $T_2$  (this character is highlighted in Figure \ref{betterworse}). So in total, $T_2$ is more parsimonious for $A_2(T_1)$ than $T_1$ (and in fact, it can even be shown that $T_2$ is a maximum parsimony tree, but not the unique one in this case, as swapping the roles of the cherries $[3,4]$ and $[5,6]$ of $T_1$ would lead to the same result). 

\begin{figure} 
\center
%\scalebox{.4}{\input{pablo8taxacounter.pstex_t} }
\scalebox{.4}{\includegraphics{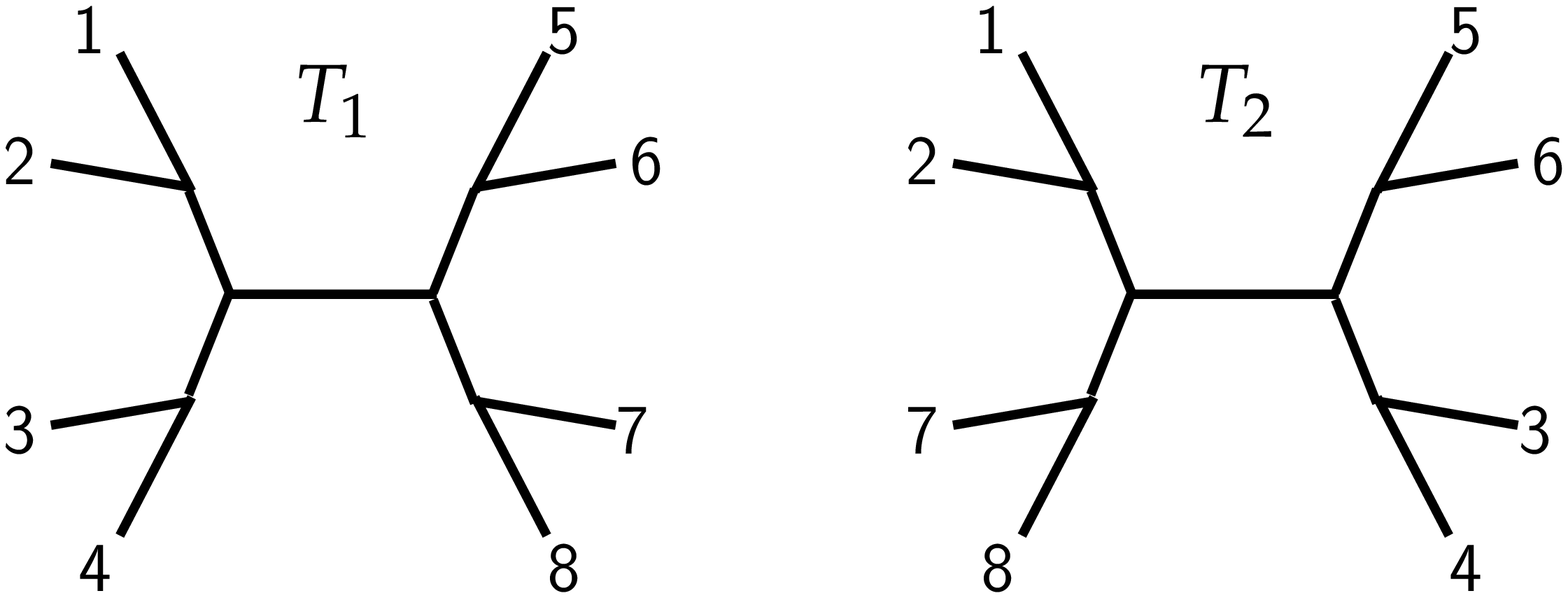} }
\caption{  \scriptsize The binary phylogenetic trees $T_1$ and $T_2$ as depicted here have the property that $l(A_2(T_1),T_2)=99$, whereas $l(A_2(T_1),T_1)=100$ (and, likewise, $l(A_2(T_2),T_2)=100$, whereas $l(A_2(T_2),T_1)=99$). This implies that $T_1$ is not a maximum parsimony tree for $A_2(T_1)$ (and neither is $T_2$ for $A_2(T_2)$). Note that the two trees differ only by swapped positions of the cherries $[3,4]$ and $[7,8]$. Alignment $A_2(T_1)$ is depicted in Figure \ref{betterworse}.}
\label{8taxacounterex}
\end{figure}

\begin{figure} 
\center
\scalebox{.65}{\includegraphics{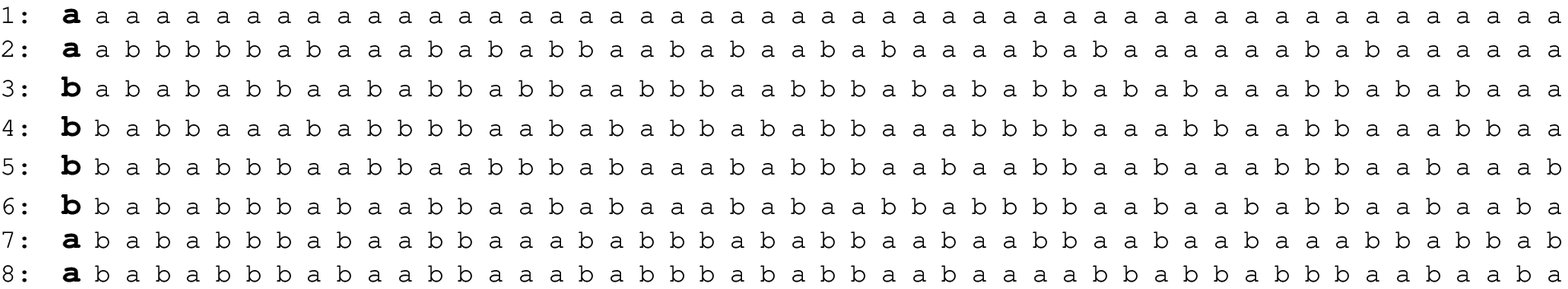}}
\caption{  \scriptsize Alignment $A_2(T_1)$ for $T_1$ as depicted in Figure \ref{8taxacounterex}. The first character is highlighted as it is the only character that has parsimony score 1 on $T_2$, whereas all others have parsimony score 2 on $T_2$ and are therefore also contained in $A_2(T_2)$. Overall, $l(A_2(T_1),T_1)=100>99=l(A_2(T_1),T_2)$. }
\label{betterworse}
\end{figure}

\section{Results}

We are now in the position to state the main result of the present manuscript.

\begin{theorem} \label{casek2} Let $T$ be a binary phylogenetic $X$-tree with $X=\{1,\ldots,n\}$, where $n\geq 9$. Then $T$ is the unique maximum parsimony tree of alignment $A_2(T)$. \end{theorem}

Note that this theorem shows that Conjecture \ref{conj} is true for $k=2$, as $n\geq9$ implies that $k<\frac{n}{4}$. It is the main aim of this manuscript to prove this theorem subsequently. However, before we can proceed with a proof of Theorem \ref{casek2}, we need to verify that the following necessary condition holds, which is a direct extension of the Buneman theorem or, more precisely, of Corollary \ref{cor:buneman}. 

\begin{proposition}\label{A2definesT} Let $T$ be a binary phylogenetic $X$-tree. Then, $A_2(T)$ defines $T$ in the sense that if $\tilde{T}$ is another binary phylogenetic $X$-tree, i.e. $T\neq \tilde{T}$, then we have $A_2(T) \neq A_2(\tilde{T})$. 
\end{proposition}

\begin{proof} If $T\neq \tilde{T}$, then $\Sigma(T)\neq \Sigma(\tilde{T})$ by Theorem \ref{thm:buneman}, but as both trees are binary, we have $|\Sigma(T)|=|\Sigma(\tilde{T})|=2n-3$ as explained in the previous section. Together, this implies that $\Sigma(T)\setminus \Sigma(\tilde{T}) \neq \emptyset$. Let $\sigma=A|B \in \Sigma(T)\setminus \Sigma(\tilde{T})$ be minimal, i.e. $\sigma=\argmin\limits_{\tilde{\sigma} \in \Sigma(T)\setminus \Sigma(\tilde{T})}|\tilde{\sigma}|$. Without loss of generality, we assume $|A|=|\sigma|$, i.e. $|A|\leq |B|$. Note that $|A|\geq 2$ as $\sigma \in \Sigma(T)$ but $\sigma\not\in \Sigma(\tilde{T})$, so $\sigma$ cannot refer to a trivial $X$-split (otherwise, it would be contained in both split sets as all $X$-trees contain edges leading to each of the leaves in $X$). Moreover, $\sigma$ divides $T$ into two subtrees $T_A$ with leaf set $A$ and $T_B$ with leaf set $B$. In the following, we denote by $T_{A_1}$ and $T_{A_2}$ the two maximal pending subtrees of $T_A$, which must exist as $|A|\geq 2$, cf. Figure \ref{treedecomp}. The taxon sets of $T_{A_1}$ and $T_{A_2}$ are denoted by $A_1$ and $A_2$, respectively. Note that $|A_1|<|A|$ and $|A_2|<|A|$, and also note that $\Sigma(T)$ must contain the two (possibly trivial) $X$-splits $\sigma_1=A_1|X\setminus A_1$ and $\sigma_2=A_2|X\setminus A_2$, as $T_{A_1}$ and $T_{A_2}$ are subtrees of $T$.

Now we define a character $f$ such that $f$ assigns all taxa in $A$ state $a$ and all taxa in $B$ state $b$ (or vice versa as we require $f(1)=a$ in this manuscript). Then, $l(f,T)=1$ as $f$ is induced by $\sigma \in \Sigma(T)$, but as $\sigma \not\in \Sigma(\tilde{T})$, we have $l(f,\tilde{T})\geq 2$. 

On the other hand, note that as $\sigma$ was chosen to be a minimal element of $\Sigma(T)\setminus \Sigma(\tilde{T})$, the $X$-splits $\sigma_1$ and $\sigma_2$ must be in $\Sigma(T)\cap\Sigma(\tilde{T})$. This is due to the fact that we already know that $\sigma_1$ and $\sigma_2$ are contained in $\Sigma(T)$, and furthermore if one of them, say $\sigma_1$, was not contained in $\Sigma(\tilde{T})$, we would have $\sigma_1\in\Sigma(T)\setminus\Sigma(\tilde{T})$ and $|\sigma_1|=|A_1|<|A|=|\sigma|$, which would contradict the minimality of $\sigma$. So we have $\sigma_1$, $\sigma_2$ $\in \Sigma(\tilde{T})$, which implies $l(f,\tilde{T})\leq 2$. In order to see this, we denote the subtree of $\tilde{T}$ containing only leaves of $A_i$ by $\tilde{T}_{A_i}$ for $i=1,2$, respectively. Now note that all inner nodes within $\tilde{T}_{A_1}$ and $\tilde{T}_{A_2}$ could be assigned state $a$, and all other inner nodes in $\tilde{T}$ could be assigned state $b$. This way, there would be precisely two changes from $a$ to $b$, namely on the two edges connecting $\tilde{T}_{A_1}$ and $\tilde{T}_{A_2}$ with the rest of tree $\tilde{T}$, respectively. 

Altogether we have $l(f,\tilde{T})\geq 2$ and $l(f,\tilde{T})\leq 2$, which implies $l(f,\tilde{T})= 2$. Thus, $f$ is contained in $A_2(\tilde{T})$, but as $l(f,T)=1$, $f$ is not contained in $A_2(T)$. Therefore, $A_2(T)\neq A_2(\tilde{T})$. This completes the proof.
\end{proof}

Proposition \ref{A2definesT} is important, as it indeed is a necessary condition for Theorem \ref{casek2}. If this proposition was not true, it would be possible for two different trees to have identical $A_2$ alignments, so in particular, both of them would have the same parsimony score. Therefore, none of them could be the unique maximum parsimony tree of this alignment. 

However, recall that the example presented in Figure \ref{8taxacounterex} shows that the fact that while $A_2(T) \neq A_2(\tilde{T})$ is necessary, it is not sufficient for Conjecture \ref{conj} to hold. 

We are finally in a position to prove Theorem \ref{casek2}. 

\begin{proof}[Proof of Theorem \ref{casek2}] We prove the statement by induction on $n$. The base case of the induction, the case $n=9$, is a direct consequence of the exhaustive search presented in \citep{pablo}. We repeated this exhaustive search in order to verify the results, so indeed, for all binary phylogenetic trees $T$ with $n=9$ leaves, we always have that $T$ is the unique maximum parsimony tree for $A_2(T)$.

So all that remains to be considered here is the inductive step. Thus, consider now a binary phylogenetic tree $T^{n+1}$ with $n+1$ leaves, where $n+1\geq 10$ (otherwise we would again be in the base case) and assume that for all binary phylogenetic trees $T$ with $n$ leaves we already know that $T$ is the unique maximum parsimony tree for $A_2(T)$. 

We now consider $A_2(T^{n+1})$. As $T^{n+1}$ has more than four leaves, we know that $T^{n+1}$ has at least two cherries \citep[Prop. 1.2.5]{Semple2003}. Our strategy is now to compare $A_2(T^{n+1})$ with the alignment $A_2(T^n)$, which shall correspond to the specific tree $T^n$ that we get when we replace one cherry by a leaf. For technical simplicity, we assume without loss of generality that one cherry is labelled $[1,n+1]$ (if this cherry does not exist in $T^{n+1}$, we re-label the leaves accordingly). So in the following, let $T^n$ be the binary phylogenetic tree that results from deleting leaf $n+1$ and the edge leading to this leaf and suppressing the resulting node of degree 2. By the inductive hypothesis, $T^n$ is the unique maximum parsimony tree for $A_2(T^n)$, and by Theorem \ref{thm:lengthAk}, using $k=2$ we know that $|A_2(T^n)|=2(n-3)^2$.

The central idea of the proof is now the following: We divide $A_2(T^{n+1})$ into two parts (subsets): Part $A$ contains only characters which assign the same state to leaves 1 and $n+1$ (as before we assume without loss of generality that leaf 1 is in state $a$, so in Part $A$, leaves 1 and $n+1$ are both in state $a$), whereas Part $B$ only contains characters which assign leaves $1$ and $n+1$ different states (i.e. without loss of generality, leaf 1 is assigned state $a$ and leaf $n+1$ is assigned state $b$). A schematic sketch of the decomposition of alignment $A_2(T^{n+1})$ is given by Figure \ref{alignmentA2}, and we will investigate this decomposition more in-depth shortly.

\begin{figure}
%\begin{tabular}{|c|||ccccc||ccccc|ccccc|} \hline \multicolumn{16}{|c|}{$A_2(T^{n+1})$}
%\\ \hline
%$1:$ & $a$ & $a$ & $\ldots$   & $\ldots$& $a$ & $a$ & $a$   & $\ldots$ & $\ldots$ & $a$ & $a$ & $a$ & $\ldots$ & $\ldots$ &  $a$ \\
%\cline{7-16}
% $2:$ &  \multicolumn{5}{|c||}{$$} 
%  &   $a$ & $a$ & $\ldots$ & $\ldots $ & $a$ &  $b$ & $b$ & $\ldots$ & $\ldots $ & $b$ \\
%  $\vdots$ &  \multicolumn{5}{|c||}{$A_2(T^n)$} 
%  &   \multicolumn{5}{c|}{$A_1(\hat{T})$} 
%  &   \multicolumn{5}{c|}{$\bar{A}_1(\hat{T})$} \\
% $n:$ &  \multicolumn{5}{|c||}{$$} 
%  &   \multicolumn{5}{c|}{$$} 
%  &   \multicolumn{5}{c|}{$$} \\ \hline
%  
%  $n+1:$ & $a$ & $a$ & $\ldots$   & $\ldots$& $a$ & $b$ & $b$   & $\ldots$ & $\ldots$ & $b$ & $b$ & $b$ & $\ldots$ & $\ldots$ &  $b$\\ \hline
%  &  \multicolumn{5}{|c||}{$A$} &  \multicolumn{10}{c|}{$B$}\\ \hline
%  
%\end{tabular}

\scalebox{.9}{\includegraphics{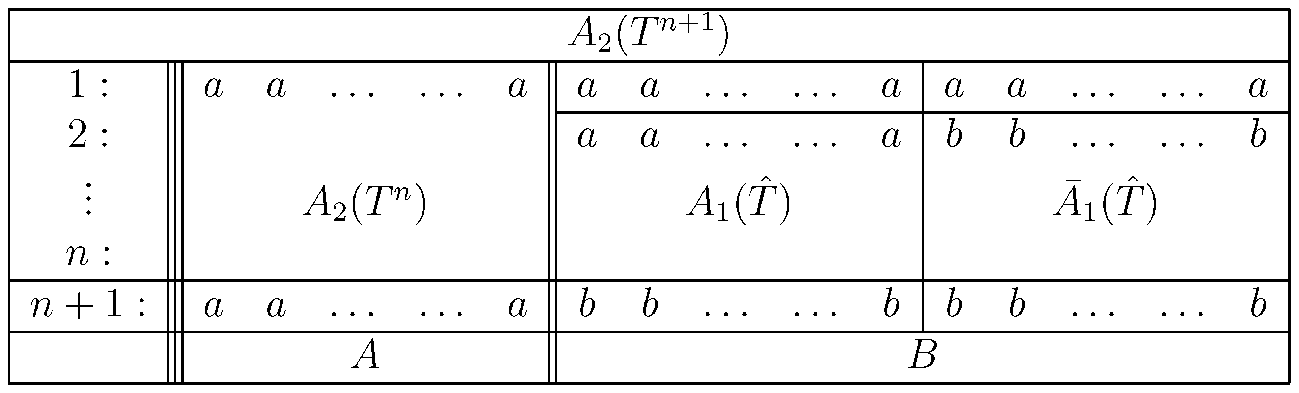}}
\caption{\scriptsize Illustration of alignment $A_2(T^{n+1})$: Part $A$ contains all characters $f$ for which $f(1)=f(n+1)=a$, and Part $B$ contains all characters $f$ for which $f(1)=a$ and $f(n+1)=b$. If we disregard $n+1$, Part $A$ corresponds to alignment $A_2(T^n)$. If we disregard both $1$ and $n+1$, alignment $B$ consists of two copies of $A_1(\hat{T})$, one of which such that taxon 2 is in state $a$ and the other one such that taxon 2 is in state $b$.  }\label{alignmentA2}
\end{figure}

We proceed as follows.

\begin{enumerate}
\item First will analyze Part $A$ of $A_2(T^{n+1})$ and show that $T^{n+1}$ is most parsimonious for this part of the alignment. In particular, we want to show that $T^{n+1}$ is the unique maximum parsimony tree for $A$, i.e. we want to show that for all binary phylogenetic trees $\tilde{T}$ on $n+1$ taxa we have $l(A,\tilde{T})\geq l(A,T^{n+1})+1$. 

We do this by showing that $A$ is closely related to $A_2(T^n)$. As a first step, we will show that $l(A,T^{n+1})=l(A_2(T^n),T^n)$. In order to see this, first note that $A$ is just like $A_2(T^n)$ but with an additional line for taxon $n+1$, which contains only $a$'s and is thus a copy of line $1$ (in particular, we have $|A|=|A_2(T^n)|=2(n-3)^2$ by Theorem \ref{thm:lengthAk}). So as in $A$ both elements of the cherry $[1,n+1]$ are in state $a$, no change on the edges of this cherry will ever happen in any most parsimonious extension of a character in $A$ (because if the node adjacent to $1$ and $n+1$ was in state $b$, there would be two changes in the cherry, but then at least one change could be saved by assigning this node state $a$ instead, as this might cause an extra change on the other edge incident to this node, but both cherry edges would then not require changes anymore). So indeed, adding $n+1$ to $T^n$ in order to get $T^{n+1}$ does not increase the parsimony score at all, so we conclude 
\begin{equation}\label{leafaddTn+1}l(A,T^{n+1})=l(A_2(T^n),T^n).\end{equation}

Now consider some other binary phylogenetic tree $\tilde{T}$ on the same leaf set as $T^{n+1}$. Again, we denote by $\tilde{T}^n$ the tree that results from deleting leaf $n+1$ from $\tilde{T}$ as well as the edge leading to $n+1$ (and suppressing the resulting node of degree 2). As can be easily seen, adding a leaf cannot decrease the parsimony score -- irregardless of the alignment and the tree. So this applies also to alignments $A_2(T^n)$ and $A$ as well as trees $\tilde{T}^n$ and $\tilde{T}$. We conclude: 

\begin{equation}\label{leafaddtildeT}l(A,\tilde{T})\geq l(A_2(T^n),\tilde{T}^n).\end{equation} 

Moreover, as by the inductive assumption $T^n$ is the unique maximum parsimony tree for $A_2(T^n)$, this implies that 

\begin{equation}\label{clear}l(A_2(T^n),\tilde{T}^n) \geq l(A_2(T^n),T^n),\end{equation} 

where equality holds if and only if $\tilde{T}^n=T^n$. 

So in total, if we summarize Equations \eqref{leafaddTn+1}, \eqref{leafaddtildeT} and \eqref{clear}, we obtain:

\begin{equation}\label{MPforA}  l(A,T^{n+1})=l(A_2(T^n),T^n) \leq l(A_2(T^n),\tilde{T}^n) \leq l(A,\tilde{T}). \end{equation}

Thus, by Equation \eqref{MPforA} we can immediately conclude that $T^{n+1}$ is a maximum parsimony tree for Part $A$ of alignment $A_2(T^{n+1})$. 

Moreover, as the inequality in Equation \eqref{clear} is strict unless $\tilde{T}^n=T^n$, for all trees $\tilde{T}$ which do \emph{not} contain $T^n$ as a subtree, we already know that their parsimony score for $A$ is strictly higher than that of $T^{n+1}$, so the first inequality in Equation \eqref{MPforA} would be strict.

If, on the other hand,  $\tilde{T}$ \emph{does } contain $T^n$ as a subtree, then it \emph{cannot} contain the cherry $[1,n+1]$, because otherwise we would have $\tilde{T}=T^{n+1}$, which would contradict the choice of $\tilde{T}$. We will now investigate this case more in-depth, as we need more details of this case later on.

\item Let us consider the case where $\tilde{T}$ does \emph{not} contain the cherry $[1,n+1]$. In particular, this means that $d_{\tilde{T}}(1,n+1)\geq 3$ (as $\tilde{T}$ is binary). Note that as $\tilde{T}$ is binary and has more than six leaves (actually, in the inductive step we may assume that $\tilde{T}$ has at least 10 taxa, as otherwise we could consider again the base case of the induction), in case $d_{\tilde{T}}(1,n+1)= 3$, there can be at most three taxa which have distance 3 to either leaf 1 or leaf $n+1$. This scenario is depicted in Figure \ref{fxontildeTNEW}. We put these up to three taxa on a list of `disregarded taxa'. 
On the other hand, in case that $d_{\tilde{T}}(1,n+1)>3$, if there is a taxon that has distance at most 3 to \emph{both} leaves 1 and $n+1$ (i.e. if it is `between' them), we put this on the list of disregarded taxa. This scenario is depicted in Figure \ref{fxontildeTNEW2}. Note that in both cases, the list of disregarded taxa prevents a `chain' of leaves $a$, $b$, $c$ in $\tilde{T}$ such that $d(a,b)$ and $d(b,c)\leq 3$ and $1$, $n+1 \in \{a,b,c\}$. Such prevented chains are depicted in Figure \ref{nochain}.

So altogether, due to $\tilde{T}$, there are now at most three taxa on the list of disregarded taxa. 

\begin{figure} 
\center
%\scalebox{.4}{ \input{pabloMikesEx.pstex_t} }
\scalebox{.4}{\includegraphics{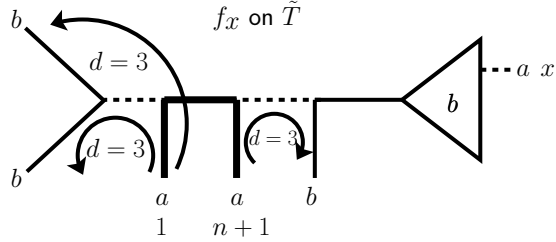}}
\caption{\scriptsize A tree $\tilde{T}$ in which leaves 1 and $n+1$ are at distance 3 to each other and can therefore be regarded as a chain of length 2 (depicted in bold). In such a tree, there can be up to 3 taxa at distance up to 3 to either one of these two leaves (note that not both 1 \emph{and} $n+1$ can be adjacent to a cherry because then the entire tree would only employ six taxa). We construct a character $f_x$ (and, analogously, $f_w$, $f_y$ and $f_z$) which assigns state $a$ to leaves 1, $n+1$ and $x$ (or $w$, $y$ or $z$, respectively) and $b$ to all other leaves. Note that $x$ (and $w$, $y$ and $z$) has a distance of more than 3 in $\tilde{T}$ to both leaves 1 and $n+1$. Therefore, the parsimony score of $f_x$ on $\tilde{T}$ is 3. The three change edges are represented by dashed lines. }
\label{fxontildeTNEW}
\end{figure}

\begin{figure} 
\center
%\scalebox{.4}{ \input{pabloMikesEx.pstex_t} }
\scalebox{.5}{\includegraphics{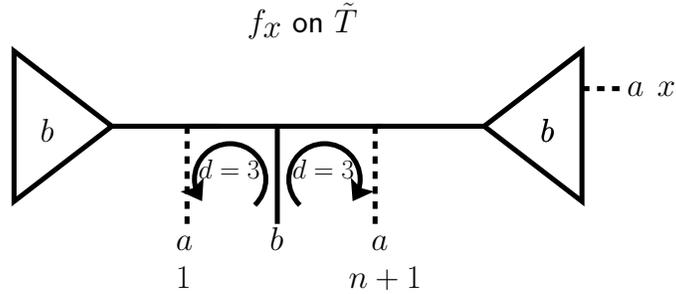}}
\caption{ \scriptsize A tree $\tilde{T}$ with $d_{\tilde{T}}(1,n+1)>3$ and with a taxon that is `in between' taxa 1 and $n+1$ in the sense that it has distance at most 3 to \emph{both} of them (note that it could also have distance 2 to one of them if it formed a cherry with this leaf). If this taxon in the middle was labelled $a$ (just as $1$ and $n+1$) by $f_w$, $f_x$, $f_y$ or $f_z$, then the parsimony score of the resulting character would be 2, not 3 (cf. Figure \ref{nochain}). Therefore, this taxon has to be labelled $b$. }
\label{fxontildeTNEW2}
\end{figure}

\begin{figure} 
\center
%\scalebox{.4}{ \input{pabloMikesEx.pstex_t} }
\scalebox{.4}{\includegraphics{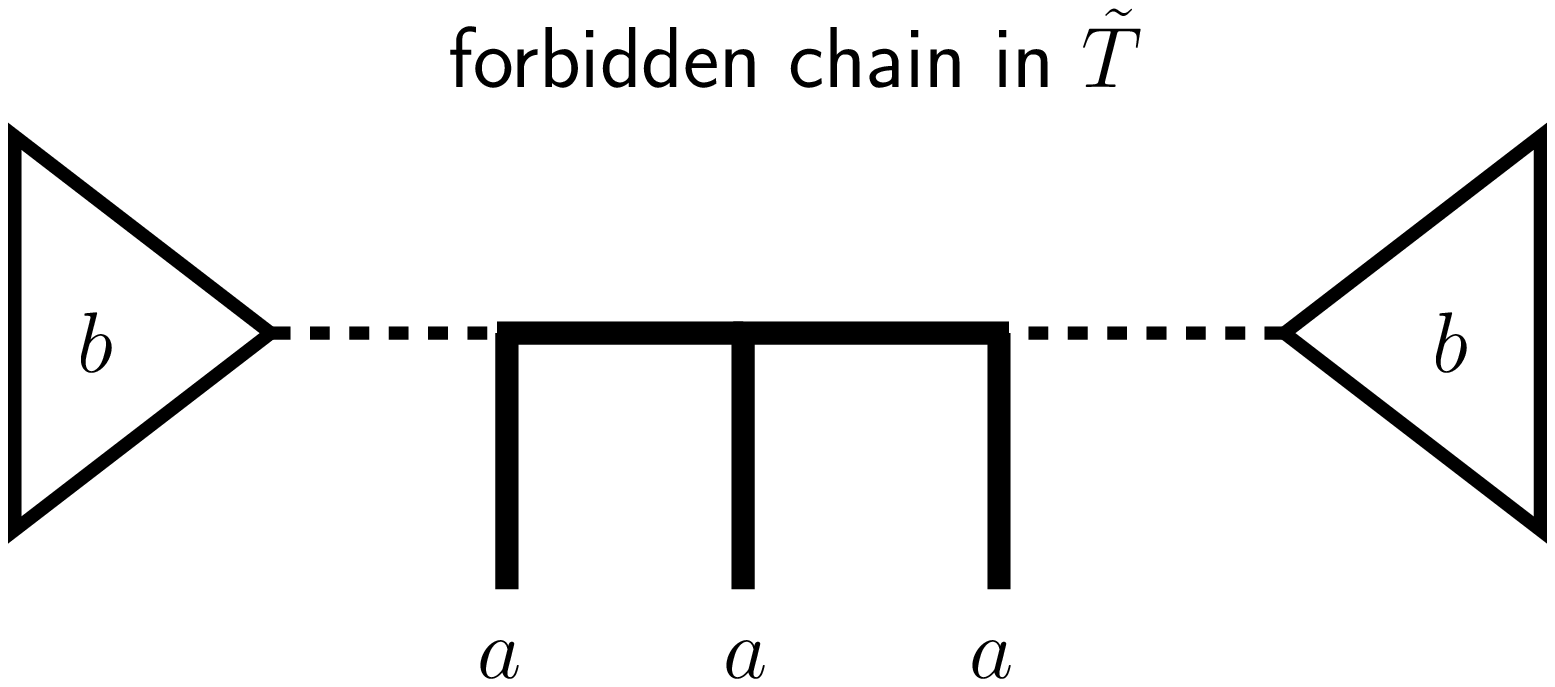}}\vspace{1cm}
\scalebox{.4}{\includegraphics{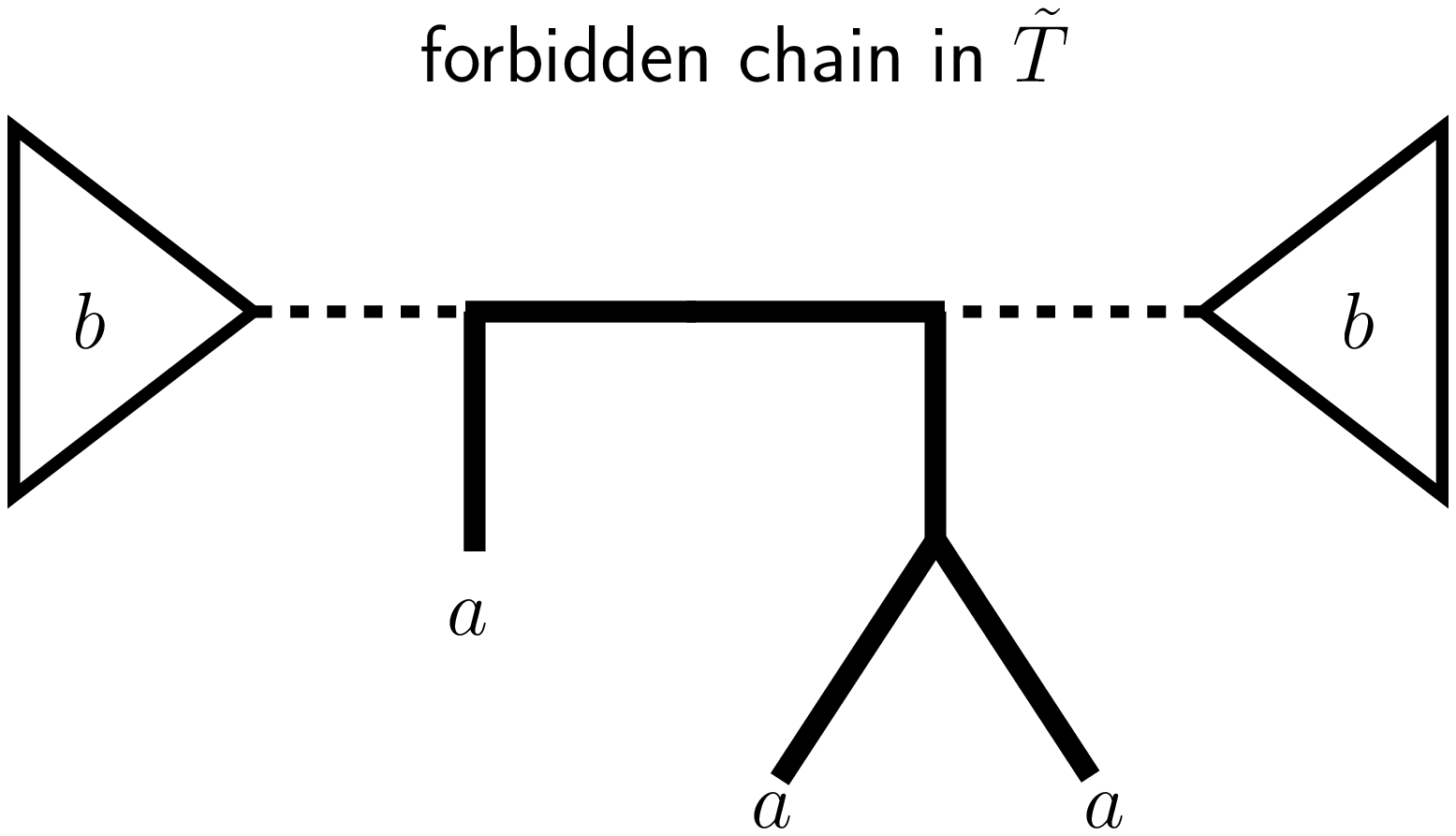}}
\caption{ \scriptsize Two types of trees $\tilde{T}$ with a chain of three taxa (cf. bold edges), i.e. with one taxon that has at most distance 3 to both of the other two taxa. If these three taxa were all labelled $a$ and all other taxa $b$ by $f_w$, $f_x$, $f_y$ or $f_z$, i.e. if two of them were leaves 1 and $n+1$ and the third one was either $w$, $x$, $y$ or $z$, then the parsimony score of the resulting character would be 2, not 3 (cf. dashed substitution edges). Therefore, in our construction of the four characters $f_w$, $f_x$, $f_y$ and $f_z$ we prevent such $a$-labelled chains. }
\label{nochain}
\end{figure}

Next consider $T^{n+1}$. As in $T^{n+1}$, 1 and $n+1$ form the cherry $[1,n+1]$ and as $T^{n+1}$ is also binary and as, again, there are at least ten leaves in total, there can be at most one taxon at distance 3 to taxon 1 (and thus also to $n+1$). This scenario is depicted in Figure \ref{fxonTn+1}. It may happen that there is such a taxon in $T^{n+1}$, and it may be that this taxon is different from the ones we already decided to disregard. We additionally disregard this taxon. Additionally, we now put taxa 1 and $n+1$ on the list of disregarded taxa.

So in total, we now disregard at most six taxa: 1, $n+1$, at most three neighbors of taxon 1 or $n+1$ in $\tilde{T}$ (i.e. taxa which have distance at most 3 to either 1 or $n+1$) and at most one distance-3 neighbor of both 1 and $n+1$ in $T^{n+1}$. If not all these taxa actually exist (e.g. if there is no distance-3 neighbor of taxon 1 in $T^{n+1}$) or if some of them coincide, we disregard fewer taxa, but the important thing is that we disregard at most six. As we have at least ten taxa in the inductive step, this means that there are at least four taxa $w$, $x$, $y$ and $z$ left, which we do \emph{not} disregard. 

For these four taxa $w$, $x$, $y$ and $z$, we now construct characters $f_w$, $f_x$, $f_y$ and $f_z$, respectively, as follows: Character $f_w$ assigns state $a$ to taxa 1, $n+1$ and $w$ and $b$ to all other taxa, character $f_x$ assigns state $a$ to taxa 1, $n+1$ and $x$, and $b$ to all other taxa. Character $f_y$ assigns state $a$ to taxa 1, $n+1$ and $y$, and $b$ to all other taxa, and, finally, character $f_z$ assigns state $a$ to taxa 1, $n+1$ and $z$, and $b$ to all other taxa.

We now investigate these four characters. In the following, we denote by $f_w^n$, $f_x^n$,  $f_y^n$ and $f_z^n$ the restrictions of $f_w$, $f_x$, $f_y$ and $f_z$ to $X\setminus\{ n+1\} = \{ 1,\ldots, n\}$. 

\begin{itemize}
\item First note that $l(f_x,T^{n+1})=2$: The parsimony score has to be at least 1 as both states $a$ and $b$ are employed, and it cannot be more than 2 because two of the three $a$'s in $f_x$ are assigned to a cherry, namely $[1,n+1]$. Moreover, the score cannot be 1 because leaf $x$ has a distance of more than 3 to $1$, so the split $1, n+1,x | X \setminus\{1,n+1,x\}$ is not contained in $T^{n+1}$. The same applies to $f_w$, $f_y$ and $f_z$. This scenario is depicted in Figure \ref{fxonTn+1}. So in total, we know that $f_w$, $f_x$, $f_y$ and $f_z$ are all contained in part $A$ of alignment $A_2(T^{n+1})$. 

\begin{figure}
\center
%\scalebox{.4}{ \input{pabloMikesEx.pstex_t} }
\scalebox{.4}{\includegraphics{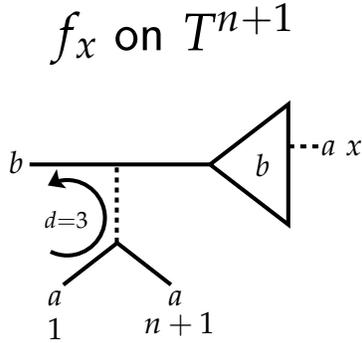}}
\caption{  \scriptsize On ${T^{n+1}}$, leaves 1 and $n+1$ form a cherry, so (as $n+1\geq 10$) they both have at most one neighbor at distance 3. Character $f_x$ (and, analogously, $f_y$ and $f_z$) assigns state $a$ to leaves 1, $n+1$ and $x$ (or $y$ or $z$, respectively) and $b$ to all other leaves. Note that $x$ (and $y$ and $z$) has a distance of more than 3 in ${T^{n+1}}$ to both leaves 1 and $n+1$, so an extra change is needed for the $a$ of leaf $x$. However, as 1 and $n+1$ form a cherry, they only require one common change. Therefore, the parsimony score of $f_x$ on $T^{n+1}$ is 2. The two change edges are represented by dashed lines. }
 \label{fxonTn+1}
\end{figure}

\item If we consider $f_w^n$, $f_x^n$, $f_y^n$ and $f_z^n$ on $\tilde{T}^n$, their respective parsimony scores can be at most 2, because $f_w^n$, $f_x^n$, $f_y^n$ and $f_z^n$ all contain only two $a$'s (so the changes could happen on the pending edges to leaves 1 and $w$, $x$, $y$ or $z$, respectively). 
\item If we consider $f_w$, $f_x$, $f_y$ and $f_z$ on $\tilde{T}$, their scores are 3. For instance, consider $f_x$: As $1$ and $n+1$ are not in a common cherry, even if we modified $f_x$ such that $x$ was in state $b$, the character would already have score 2, because $\tilde{T}$ does not contain the split $1,n+1 | X \setminus \{1,n+1\}$. But as $f_x$ assigns $x$ state $a$ and as $x$, $1$ and $n+1$ do not form a `chain' of length 3 (see above) in $\tilde{T}$, the parsimony score of $f_x$ is 3. Note that there are two cases: Either $x$, $1$ and $n+1$ all have distance more than 3 to each other, in which case the 3 substitutions can happen on their respective pending edges (cf. Figure \ref{fxontildeTNEW2}), or we have one `chain' of length 2 and one individual taxon (cf. Figure \ref{fxontildeTNEW}). However, both cases lead to a parsimony score of 3. The cases of $f_w$, $f_y$ and $f_z$ are analogous to that of $f_x$.

\end{itemize}

So in summary, we have found four characters, namely $f_w$, $f_x$, $f_y$ and $f_z$ in alignment part $A$, which have score at most 2 (actually even precisely 2, but this is not important here) on $\tilde{T}^n$, but have score 3 on $\tilde{T}$. Together with the fact that attaching an extra leaf cannot decrease the score of any character on any tree, this implies that attaching leaf $n+1$ strictly increases the score of alignment $A$ by at least 4 on any tree $\tilde{T}$ which does not put 1 and $n+1$ together in a cherry. Basically, this is due to the fact that the sequence of $n+1$ in $A$ is a precise copy of the sequence of taxon 1 (both are in state $a$ for all characters in $A$), which is why it cannot be optimal for alignment $A$ to separate taxa 1 and $n+1$.

Anyway, in summary, we now know for any tree $\tilde{T}$ which does not contain the cherry $[1,n+1]$:

\begin{equation}\label{+4forA} l(A,\tilde{T})\geq l(A_2(T^n),\tilde{T}^n)+4,\end{equation}
 
where the 4 is due to the fact that the score of no character can improve when leaf $n+1$ is attached, but the scores of at least four characters (namely $f_w$, $f_x$, $f_y$ and $f_z$) strictly get worse.

\item Next we consider Part $B$ of $A_2(T^{n+1})$, i.e. the part where leaves 1 and $n+1$ are in different states. We will show that $T^{n+1}$ is surprisingly {\em not} a maximum parsimony tree for this part of the alignment, but that its parsimony score differs only by 2 from the optimal score.

Note that alignment $B$ is such that leaf 1 is assigned state $a$, whereas leaf $n+1$ is assigned state $b$, which means that for each character $f \in B$, the cherry $[1,n+1]$ of $T^{n+1}$ contributes 1 to the parsimony score of 2 (we know $l(f,T^{n+1})=2$ because $f \in B\subseteq A_2(T^{n+1})$). This implies that the rest of the tree contributes precisely one change to the parsimony score of $f$ on $T^{n+1}$ for each $f\in B$. Consider such a character $f$ and a most parsimonious extension of $f$ on $T^{n+1}$. Then, there is one change on the cherry $[1,n+1]$, and there cannot be a change on the edge leading to this cherry. This is due to the fact that if there was such a change on the third edge incident to the node $u$ adjacent to leaves 1 and $n+1$, it would save a change to assign $u$ the other one of the two states (the cherry would still require only one change). So relabeling $u$ would save a change, which means that the original extension would not have been most parsimonious, which would be a contradiction.

We now consider the tree $\hat{T}$ with $n-1$ leaves which results from $T^{n+1}$ when we delete cherry $[1,n+1]$ and node $u$ as well as all edges incident with $u$ and if we suppress the resulting node of degree 2. Note that for each character $f$ in $B$, as the cherry $[1,n+1]$ contributes one change to the parsimony score of $f$, $\hat{T}$ also contributes exactly one change to each most parsimonious extension of $f$ on $T^{n+1}$. Thus, if we disregard cherry $[1,n+1]$, each such character $f$ corresponds to an edge of $\hat{T}$ (and thus a split of the taxon set $\{2,\ldots,n\}$). We will now show that in turn, every edge of $\hat{T}$ corresponds to \emph{two} characters in $B$. 

Therefore, consider alignment $\hat{B}$ which results from $B$ by deleting lines 1 and $n+1$. Note that in this alignment, all splits induced by the characters appear exactly twice. This is due to the fact that a split into two disjoint subsets $X_1$ and $X_2$ of taxon set $\{2,\ldots,n\}$ can be such that the taxa in $X_1$ are assigned $a$ (and thus are grouped together with leaf 1) and the taxa in $X_2$ are assigned $b$ (and thus are grouped together with leaf $n+1$) or vice versa. So for each character $f \in B$ there is a character $\bar{f} \in B$ such that the roles of $X_1$ and $X_2$ are interchanged, and we thus have to consider all edges of the binary phylogenetic tree $\hat{T}$ on taxon set $\{2,\ldots,n\}$ twice. We denote by $A_1(\hat{T})$ the part of this alignment $\hat{B}$ which contains all splits of $\hat{T}$ such that taxon 2 is in state $a$, and by $\bar{A}_1(\hat{T})$ the part of $\hat{B}$ which contains all splits of $\hat{T}$ such that taxon 2 is in state $b$. Consider again Figure \ref{alignmentA2} for clarification of this decomposition.\footnote{Note that the decomposition can also be verified enumeratively: We know that $|A|=|A_2(T_n)|$, and by Theorem \ref{thm:lengthAk}, we can thus conclude that  $|A|=2(n-3)^2$. Moreover, any binary tree with $n-1$ leaves has $2(n-1)-3=2n-5$ edges, so this applies also to $\hat{T}$. Therefore, we derive $2\cdot |A_1(\hat{T})|=|\hat{B}|=|B|=4n-10$. So in total, we get $|A_2(T^{n+1})|=|A|+|B|= 2n^2-8n+8=2((n+1)-3)^2$. This in turn equals $|A_2(T^{n+1})|$ by Theorem \ref{thm:lengthAk}, which confirms this observation.}

We now distinguish two cases. \begin{enumerate} \item We first consider any binary phylogenetic tree $\tilde{T}$ on taxon set $\{1,\ldots,n+1\}$ which contains the cherry $[1,n+1]$. As leaves 1 and $n+1$ are in different states for all characters $f \in B$, this means that the cherry contributes 1 to the parsimony score of $f$ on $\tilde{T}$ for all $f \in B$. In particular, deleting the cherry from $f$ in order to obtain a character $\hat{f} \in \hat{B}$ would strictly decrease the parsimony score by 1. 
However, as $l(\hat{f},T) \geq 1$ for all binary phylogenetic trees $T$ on taxon set $\{2,\ldots,n\}$ (as otherwise $\hat{f}$ would be constant and could not correspond to an edge of $\hat{T}$), this leads to $l(f,\tilde{T})=l(\hat{f},\tilde{T}|_{\{2,\ldots,n\}}) +1 \geq 1+1=2$. Here, $  \tilde{T}|_{\{2,\ldots,n\}}$ denotes the tree resulting from $\tilde{T}$ when cherry $[1,n+1]$ as well as its parent node $u$ are removed and the resulting degree 2 node is suppressed. In particular, this leads to $l(f,\tilde{T})\geq 2=l(f,T^{n+1})$. As this holds for all characters $f \in B$,  for such a tree $\tilde{T}$ we have $l(B,\tilde{T})\geq l(B,T^{n+1})$.

\item Next we consider any binary phylogenetic tree $\tilde{T}$ on taxon set $\{1,\ldots,n+1\}$ which does \emph{not} contain the cherry $[1,n+1]$.

Consider any character $f \in B$, and denote by $\hat{f}$ the version of $f$ in which leaves 1 and $n+1$ have been deleted. Then, as $\hat{f}$ already corresponds to a split of $\hat{T}$ (as explained above), we know that $l(f,\tilde{T})\geq 1$, and $f$ contains both $a$ and $b$ more than once (because even without considering leaves 1 and $n+1$, $f$ already induced a split and was thus not constant). So if we want to consider the difference between $l(B,\tilde{T})$ and $l(B,T^{n+1})$, we know that the score of each individual character $f \in B$, which is 2 on $T^{n+1}$ (as $B \subseteq A_2(T^{n+1})$), can improve by at most 1 for any other tree (as we can only go down from 2 to 1, but not to 0). Moreover, in order to achieve this, $f$ has to correspond to an edge $e_f=\{u,v\}$ of $\tilde{T}$. 

So now assume we have such a character $f \in B$ that corresponds to an edge of $\tilde{T}$ and contains at least two $a$'s and at least two $b$'s, and with leaf 1 in state $a$ and leaf $n+1$ in state $b$. Then, the edge $e_f=\{u,v\}$ corresponding to $f$ in $\tilde{T}$ must separate leaf 1 from leaf $n+1$, i.e. it must lie on the path from 1 to $n+1$ in $\tilde{T}$ (otherwise, more than one change would be required). Without loss of generality, we assume that $u$ is closer to leaf 1 than $v$ (measured in terms of the number of edges on the path connecting these nodes), i.e. the most parsimonious extension of $f$ on $\tilde{T}$ would assign $u$ state $a$ and $v$ state $b$. Note that neither $u$ nor $v$ can be equal to leaves 1 or $n+1$, because this would imply that $e_f$ coincides with an edge incident to either leaf 1 or leaf $n+1$. But this cannot happen, because then only one $a$ or only one $b$ would be split from the rest, but we know we have more than one of each state in $f$.

We now distinguish two cases:

\begin{itemize}
\item $e_f=\{u,v\}$ is such that $u$ is not adjacent to leaf 1 \emph{and} $v$ is not adjacent to leaf $n+1$, or
\item $u$ is adjacent to leaf 1 or $v$ is adjacent to leaf $n+1$ (or both).
\end{itemize} 

If neither $u$ is adjacent to leaf 1 nor $v$ to leaf $n+1$, then such a character $f$ on $\tilde{T}$ looks as depicted in Figure \ref{unproblematicCase1}. Assume without loss of generality that $f$ lies in the part of $B$ that corresponds to $A_1(\hat{T})$, and consider its corresponding character $\bar{f}$ in the part of $B$ corresponding to $\bar{A}_1(\hat{T})$, i.e. the character which, when 1 and $n+1$ are deleted, induces the same split as $f$ but has the roles of $a$ and $b$ reversed. 

\begin{figure} 
\center
%\scalebox{.4}{ \input{pabloMikesEx.pstex_t} }
\scalebox{.4}{\includegraphics{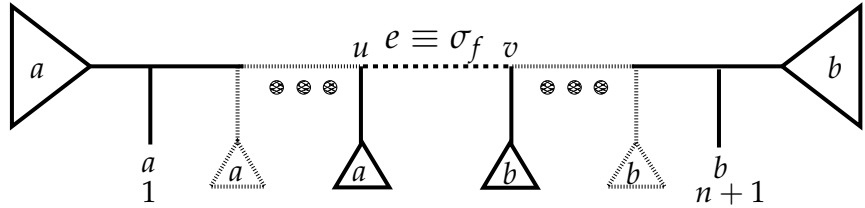}}
\caption{  \scriptsize A tree $\tilde{T}$ which contains an edge $e=\{u,v\}$ corresponding to split $\sigma_f$ induced by a character $f\in B$, but such that $u$ and $v$ are both neither adjacent to leaf 1 or leaf $n+1$. In this case, we have $l(f,\tilde{T})=1$ (as a change only needs to happen on $e$, which is highlighted by the dashed edge), but -- as is shown in Figure \ref{unproblematicCase2} -- in this case, $\bar{f}$ has a parsimony score of 3 on $\tilde{T}$.}
\label{unproblematicCase1}
\end{figure}

Now if $f$  on $\tilde{T}$ looks as depicted in Figure \ref{unproblematicCase1}, i.e. if $f$ induces an edge on the path from 1 to $n+1$ whose endpoints are neither adjacent to 1 nor to $n+1$, then $\bar{f}$ has a parsimony score of 3, as depicted in Figure \ref{unproblematicCase2}. Therefore, $l(f,\tilde{T})+l(\bar{f},\tilde{T})=1+3=4=2+2=l(f,T^{n+1})+l(\bar{f},T^{n+1})$. Thus, such a character $f$ in $B$ which decreases the score by 1 (compared to the score on $T^{n+1}$) comes paired with a character $\bar{f}$ also in $B$, which increases the score by 1, which means that there is no net difference between $\tilde{T}$ and $T^{n+1}$ concerning the sum of parsimony scores of $f$ and $\bar{f}$.

\begin{figure} 
\center
%\scalebox{.4}{ \input{pabloMikesEx.pstex_t} }
\scalebox{.4}{\includegraphics{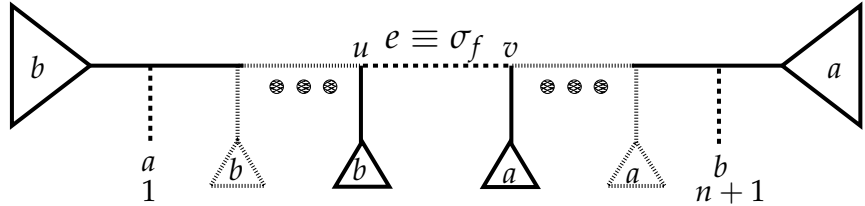}}
\caption{  \scriptsize Tree $\tilde{T}$ from Figure \ref{unproblematicCase1} with character $\bar{f}$. Here, we have $l(\bar{f},\tilde{T})=3$ (as we need changes on $e$ as well as the edges leading to leaves 1 and $n+1$, respectively -- these edges are dashed to show the changes). So, even given that -- as is shown in Figure \ref{unproblematicCase1} -- in this case, $f$ only has a parsimony score of 1 on $\tilde{T}$, the net impact of $f$ and $\bar{f}$ on $l(B,\tilde{T})$ is $1+3=4$, i.e. it is the same as on $T^{n+1}$, where it is $2+2=4$.}
\label{unproblematicCase2}
\end{figure}

If, on the other hand, $u$ is adjacent to leaf 1, then such a character $f$ looks on $\tilde{T}$ as depicted in Figure \ref{problematicCase1}, and its corresponding character $\bar{f}$ would come with a parsimony score of 2 as depicted in Figure \ref{problematicCase2}. So in total, we would have $l(f,\tilde{T})+l(\bar{f},\tilde{T})=1+2=3<4=2+2=l(f,T^{n+1})+l(\bar{f},T^{n+1})$. If $v$ is adjacent to $n+1$, the scenario is analogous. 
\end{enumerate}

\begin{figure} 
\center
%\scalebox{.4}{ \input{pabloMikesEx.pstex_t} }
\scalebox{.4}{\includegraphics{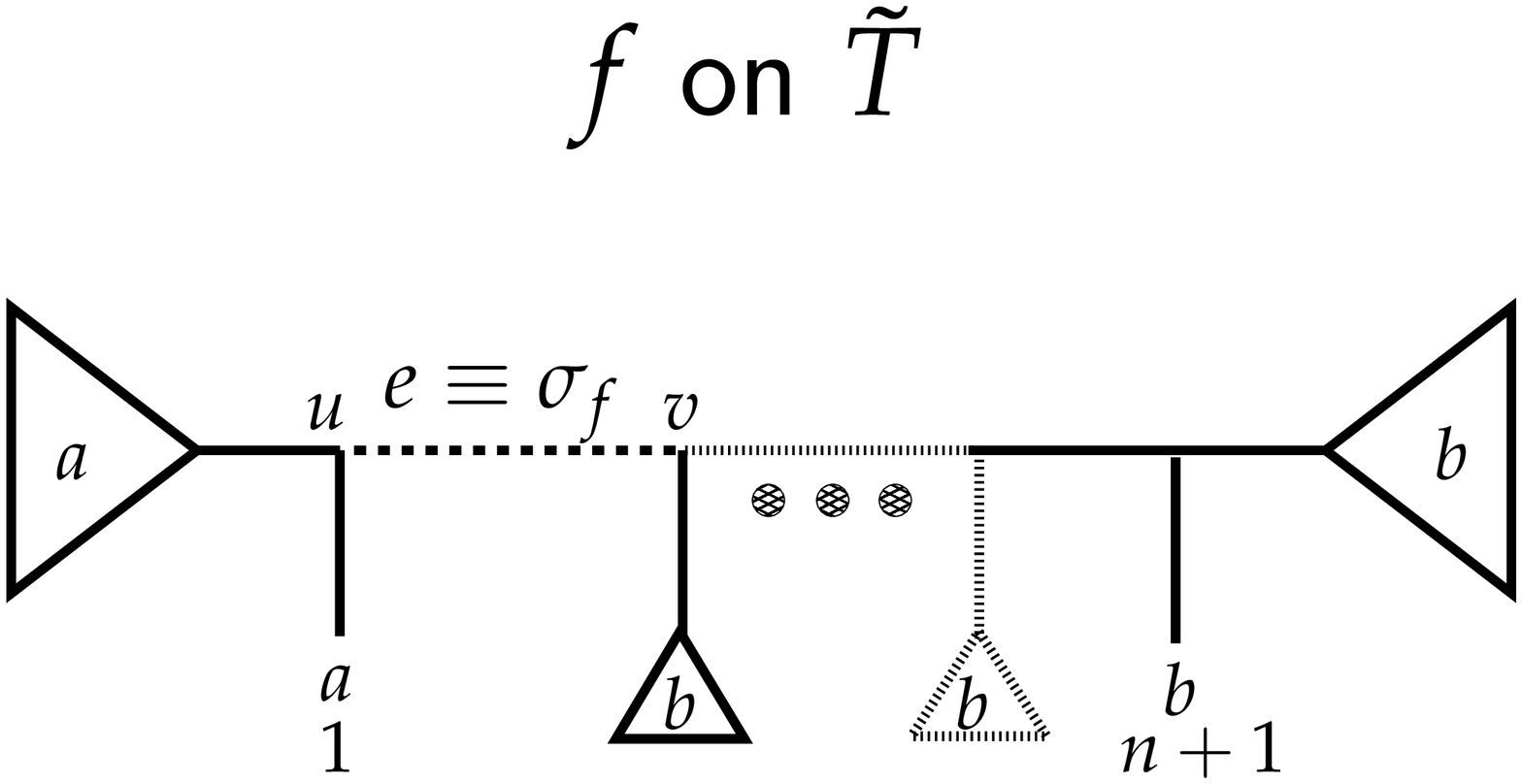}}
\caption{  \scriptsize A tree $\tilde{T}$ which contains an edge $e=\{u,v\}$ corresponding to split $\sigma_f$ induced by a character $f\in B$, such that $u$ is adjacent to leaf 1 (or, analogously, $v$ is adjacent to leaf $n+1$). In this case, we have $l(f,\tilde{T})=1$ (as a change only needs to happen on $e$, which is highlighted by the dashed edge), and -- as is shown in Figure \ref{problematicCase2} -- in this case, $\bar{f}$ has a parsimony score of 2 on $\tilde{T}$.}
\label{problematicCase1}
\end{figure}

\begin{figure} 
\center
%\scalebox{.4}{ \input{pabloMikesEx.pstex_t} }
\scalebox{.4}{\includegraphics{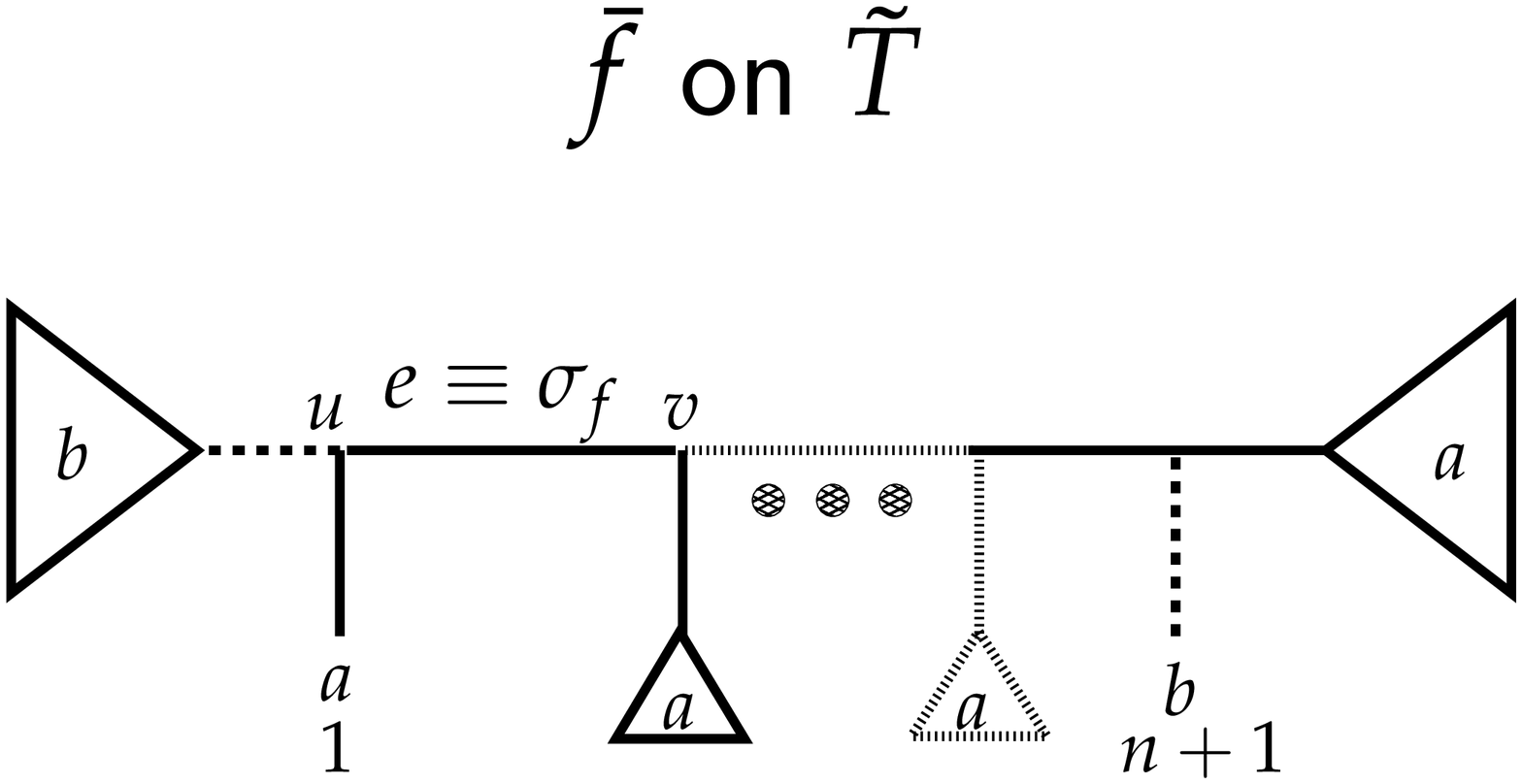}}
\caption{  \scriptsize Tree $\tilde{T}$ from Figure \ref{problematicCase1} with character $\bar{f}$. Here, we have $l(\bar{f},\tilde{T})=2$ (as we need changes on the third edge incident to $u$, i.e. other than $e$ and the edge leading to leaf 1, as well as on the edge leading to leaf $n+1$ -- these edges are dashed to show the changes). So, considering given that -- as is shown in Figure \ref{problematicCase1} -- in this case, $f$ only has a parsimony score of 1 on $\tilde{T}$, the net impact of $f$ and $\bar{f}$ on $l(B,\tilde{T})$ is $1+2=3$, which is strictly less than on $T^{n+1}$, where it is $2+2=4$.}
\label{problematicCase2}
\end{figure}

So in summary, depending on the positions of leaves 1 and $n+1$ in $\tilde{T}$, we conclude for $l(B,\tilde{T})$:

\begin{itemize}
\item If $\tilde{T}$ contains the cherry $[1,n+1]$, then $l(B,\tilde{T}) \geq l(B,T^{n+1})$.
\item If $\tilde{T}$ contains precisely one inner edge $e=\{u,v\}$ on the path from leaf 1 to leaf $n+1$, then $l(B,\tilde{T}) \geq l(B,T^{n+1})-1$ (because then, the only pair of characters $f$, $\bar{f}$ that can improve the score of $B$ by 1 corresponds to the case where $u$ is adjacent to leaf 1 and at the same time $v$ is adjacent to leaf $n+1$ or vice versa). 
\item If $\tilde{T}$ contains more than one inner edge on the path from leaf 1 to leaf $n+1$, then $l(B,\tilde{T}) \geq l(B,T^{n+1})-2$, as then there are two edges on the path from 1 to $n+1$ which induce splits $f_1$ and $f_2$, respectively, for which $l(f_1,\tilde{T})=l(f_2,\tilde{T})=1$ and $l(\bar{f}_1,\tilde{T})=l(\bar{f}_2,\tilde{T})=2$. 
\end{itemize}

\item We now summarize our results. Let $\tilde{T} \neq T^{n+1}$ be any binary phylogenetic tree on the same taxon set as $T^{n+1}$. Then, there are two cases: either $\tilde{T}$ contains cherry $[1,n+1]$ or not. 

\begin{itemize}
\item If $\tilde{T}$ contains cherry $[1,n+1]$, then it cannot contain $T^n$ as a subtree (because otherwise $\tilde{T}=T^{n+1}$), and thus, Inequality \eqref{clear} is strict, which is why, by Equation \eqref{MPforA} we have 

\begin{equation*} l(A,\tilde{T})>l(A,T^{n+1}).\end{equation*}

Moreover, if $\tilde{T}$ contains cherry $[1,n+1]$, then we have seen that\begin{equation*}l(B,\tilde{T})\geq l(B,T^{n+1}).\end{equation*}

So, in summary, if  $\tilde{T}$ contains cherry $[1,n+1]$, then (as $A_2(T^{n+1})=A.B$), we get 

\begin{eqnarray}\nonumber l(A_2(T^{n+1}),\tilde{T})&=&l(A,\tilde{T})+l(B,\tilde{T})\\ \nonumber &>& l(A,T^{n+1})+l(B,T^{n+1})\\ \nonumber &=& l(A_2(T^{n+1}),T^{n+1}).\end{eqnarray}

So the parsimony score of $A_2(T^{n+1})$ on $T^{n+1}$ is strictly smaller than on any other tree $\tilde{T}$ which also contains cherry $[1,n+1]$.

\item If, on the other hand, $\tilde{T}$ does not contain cherry $[1,n+1]$, we have shown that 
\begin{equation*} l(A,\tilde{T})\geq l(A_2(T^n),\tilde{T}^n)+4.\end{equation*}

Together with Equation \eqref{MPforA}, this leads to 

\begin{equation*} l(A,\tilde{T})\geq l(A,T^{n+1})+4.\end{equation*}

Moreover, for such trees we have seen that 
\begin{equation*} l(B,\tilde{T})\geq l(B,T^{n+1})-2.\end{equation*}

So, in summary, if  $\tilde{T}$ does not contain cherry $[1,n+1]$, then we get

\begin{eqnarray*} l(A_2(T^{n+1}),\tilde{T})&=&l(A,\tilde{T})+l(B,\tilde{T})\\&\geq& l(A,T^{n+1})+4+l(B,T^{n+1}) -2\\&=& l(A_2(T^{n+1}),T^{n+1})+2.\end{eqnarray*}

So the parsimony score of $A_2(T^{n+1})$ on $T^{n+1}$ is strictly smaller than on any other tree $\tilde{T}$ which does not contain the cherry $[1,n+1]$.
\end{itemize}

Therefore, in both cases (whether $\tilde{T}$ contains the cherry $[1,n+1]$ or not), we conclude that the parsimony score of $A_2(T^{n+1})$ on $\tilde{T}$ is strictly larger than that of $T^{n+1}$, which implies that $T^{n+1}$ is indeed the unique maximum parsimony tree of alignment $A_2(T^{n+1})$. This completes the proof of Theorem \ref{casek2}.

\end{enumerate}
\end{proof}

We need one last lemma before we can conclude this section with our final result.

\begin{lemma} \label{concatlemma} Let $T$ be a binary phylogenetic $X$-tree such that $T$ is a maximum parsimony tree of alignments $A$ and $B$ and for one of them even unique with this property. Then, $T$ is also the unique maximum parsimony tree of the concatenated alignment $A.B$.
\end{lemma}

\begin{proof} Assume without loss of generality that $T$ is the unique maximum parsimony tree for alignment $A$, and it is also a maximum parsimony tree of alignment $B$ (not necessarily unique). Then we have for all binary phylogenetic $X$-trees $\tilde{T}$ with $\tilde{T}\neq T$: $l(A,\tilde{T})>l(A,T)$ and $l(B,\tilde{T})\geq l(B,T)$. Moreover, as the parsimony score of an alignment by definition is just the sum of the parsimony scores of its characters, it can be a easily seen that $l(A.B,T)=l(A,T)+l(B,T)$ and $l(A.B,\tilde{T})=l(A,\tilde{T})+l(B,\tilde{T})$. Therefore, we get for all binary phylogenetic $X$-trees $\tilde{T}$ with $\tilde{T}\neq T$:
$$ l(A.B,\tilde{T}) =\underbrace{ l(A,\tilde{T})}_{>l(A,T)}+\underbrace{l(B,\tilde{T})}_{\geq l(B,T)} > l(A,T)+l(B,T) = l(A.B,T).$$ Thus, $T$ is the unique maximum parsimony tree for $A.B$. This completes the proof.
\end{proof}

We end this section with the following corollary, which together with Lemma \ref{concatlemma} generalizes Theorem \ref{casek2}, which applies only to $k=2$, to the case with $k\leq 2$. 

\begin{corollary}[Generalization of Theorem \ref{casek2}] \label{concat} Let $T$ be a binary phylogenetic $X$-tree with $|X|\geq 9$. Then, $T$ is the unique maximum parsimony tree for the alignments $A_0.A_1(T)$, $A_0.A_2(T)$, $A_1(T).A_2(T)$ and $A_0.A_1(T).A_2(T)$.
\end{corollary}

\begin{proof} For $A_0$, which consists only of the constant character $f=a \ldots a$, actually all binary phylogenetic trees are maximum parsimony trees, because this character does not require a change on any tree. So clearly, $T$ is a maximum parsimony tree for $A_0$ (but $T$ is not unique with this property). By Corollary \ref{cor:buneman}, $T$ is the unique maximum parsimony tree for $A_1(T)$, and by Theorem \ref{casek2}, $T$ is also the unique maximum parsimony tree for $A_2(T)$. So by Lemma \ref{concatlemma}, $T$ is the unique maximum parsimony tree for all concatenations stated above. This completes the proof.
\end{proof}

We end this section by pointing out that if an alignment contains any of the concatenated alignments of Corollary \ref{concat} as well as additional copies of any characters that are contained $A_0$, $A_1(T)$ or $A_2(T)$, $T$ will still be the unique most parsimonious tree.\footnote{This corresponds to turning $A_0$, $A_1(T)$ or $A_2(T)$ into multisets rather than sets.} This implies that if we have such alignments with only up to two changes per character, maximum parsimony will recover the correct tree -- which justifies the usage of maximum parsimony in such instances. However, note that for $A_1(T)$, uniqueness will get lost when fewer characters are considered, i.e. when not all characters of $A_1(T)$ are present. Moreover, for $A_2(T)$, additionally the property of being a maximum parsimony tree might be lost if not \emph{all} characters of the alignment are there. For instance, as the proof of Theorem \ref{casek2} shows, if only part $B$ of alignment $A_2(T)$ is considered, there are other trees that have a strictly better score.

\section{Discussion and Outlook}
It was the main aim of this manuscript to prove the special case of $k=2$ of Conjecture \ref{conj}, because this conjecture is of both mathematical and biological interest as maximum parsimony is often assumed to be justified for evolutionary tree estimation  when the number of changes is small \citep[Chapter 5]{Semple2003}). We have shown that for this conjecture to hold, we require a Buneman-type necessary condition for the $A_2$ alignments, namely that they are unique for each tree (cf. Proposition \ref{A2definesT}). While we carefully analyzed our main theorem for potentially shorter proofs (e.g. using a variation of Menger's theorem (cf. for instance \citep[Lemma 5.1.7 and Corollary 5.1.8]{Semple2003}) or the Erd\"os-Sz\'ekely theorem (cf. \citep[Theorems 3 and 4]{ErdosSzekely})), the proofs given in this manuscript are the most concise ones we could achieve. We conjecture, though, that it might be possible to exploit either one of the mentioned theorems or some other combinatorial properties in order to derive a shorter proof. 

Moreover, of course our proof that Conjecture \ref{conj} holds (cf. Theorem \ref{casek2}) when $k=2$ is only a first step towards proving (or disproving) the conjecture for all values of $k$ with $k<\frac{n}{4}$, and more research is needed to tackle this general case. A particular difficulty arises from the fact that the Buneman-type necessary condition does no longer hold in general when $k>2$. For instance, consider the two trees $T_1$ and $T_2$ depicted in Figure \ref{A3problem} and their respective $A_3$ alignments. It can be easily verified that $A_3(T_1)=A_3(T_2)$, even though $T_1\neq T_2$. Alignment $A_3(T_1)=A_3(T_2)$ is depicted in Figure \ref{problemal6taxak3}. 

\begin{figure} 
\center
%\scalebox{.4}{ \input{pablo6taxaA3prob.pstex_t} }
\scalebox{.4}{\includegraphics{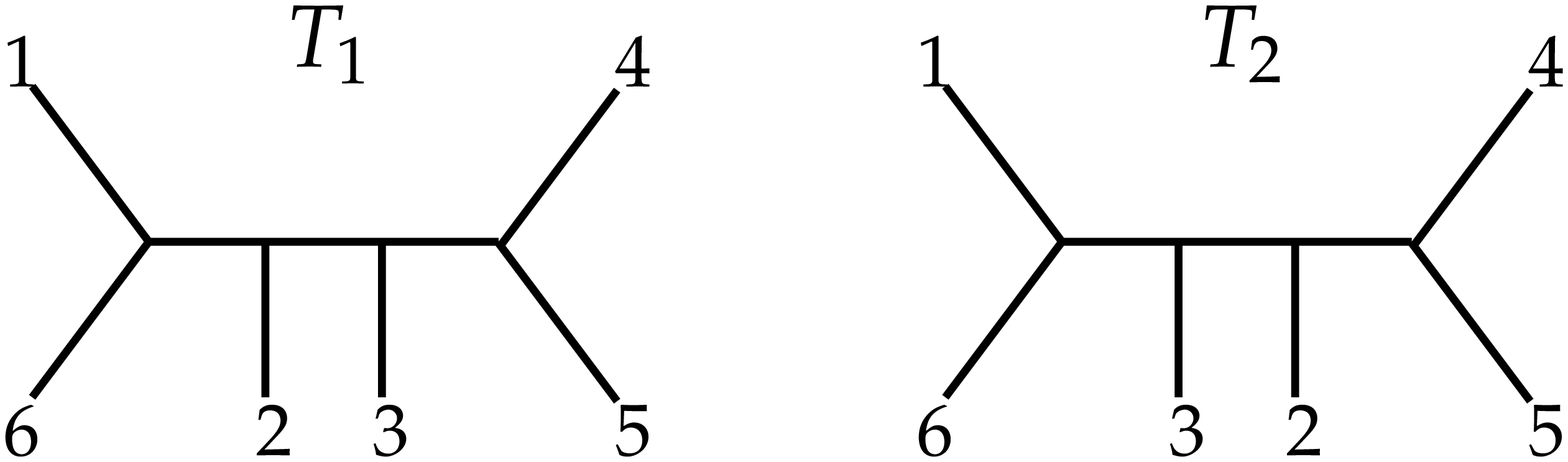}}
\caption{  \scriptsize Trees $T_1$ and $T_2$ which share the same $A_3$ alignment, which is depicted in Figure \ref{problemal6taxak3}.  }
\label{A3problem}
\end{figure}

\begin{figure} 
\center
%\scalebox{.29}{ \input{hereditary7a.pstex_t} }
\scalebox{1}{\includegraphics{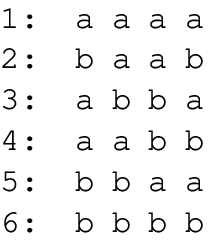}}
\caption{  \scriptsize Alignment $A_3(T_1)=A_3(T_2)$ for $T_1$ and $T_2$ as depicted in Figure \ref{A3problem}. As the two trees lead to identical $A_3$ alignments, it is impossible for maximum parsimony (or any other tree reconstruction criterion) to tell the two trees apart based on this alignment.  }
\label{problemal6taxak3}
\end{figure}

In fact, as depicted in Figures \ref{Anhalfproblem} and \ref{Anhalfproblem2}, the $A_3$ example can be generalized to all values of $k\geq 3$:

\begin{figure} 
\center
%\scalebox{.4}{ \input{pabloMikesEx.pstex_t} }
\scalebox{.3}{\includegraphics{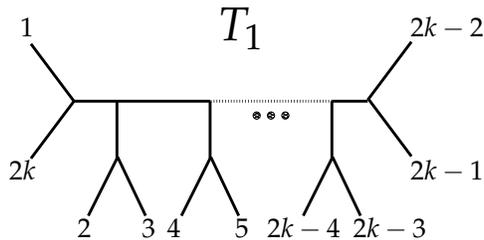}}
\caption{  \scriptsize Tree $T_1$ with $n=2k$ leaves, which has the same $A_k$ alignment as $T_2$, which is depicted in Figure \ref{Anhalfproblem2}. }
\label{Anhalfproblem}
\end{figure}

\begin{figure} 
\center
%\scalebox{.4}{ \input{pabloMikesEx2.pstex_t} }
\scalebox{.3}{\includegraphics{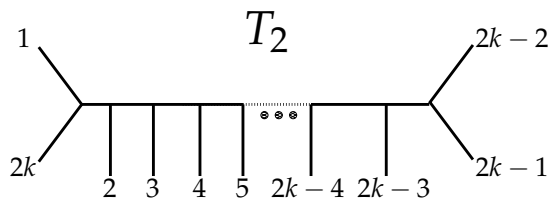}}
\caption{  \scriptsize Tree $T_2$ with $n=2k$ leaves, which has the same $A_k$ alignment as $T_1$, which is depicted in Figure \ref{Anhalfproblem}. }
\label{Anhalfproblem2}
\end{figure}

\begin{problem} \label{prob1} Let $k\geq 3$. Then there exist two binary phylogenetic trees $T_1$ and $T_2$ on taxon set $X=\{1,\ldots,n=2k\}$ such that $T_1\neq T_2$, but $A_k(T_1)=A_k(T_2)$. 
\end{problem}

A short proof that $T_1$ and $T_2$ from Figures \ref{Anhalfproblem} and \ref{Anhalfproblem2} indeed share the same $A_k$ alignment is given in the appendix.

However, note that this does not disprove Conjecture \ref{conj}, as the problematic example stated here requires $n=2k$ and thus $k=\frac{n}{2}$, but the conjecture actually requires $k<\frac{n}{4}$. So in order to prove Conjecture \ref{conj} for $k>3$, one will first have to prove that when $k<\frac{n}{4}$, we have that $T_1\neq T_2$ implies $A_k(T_1)\neq A_k(T_2)$. It is remarkable that for the case $k=2$, this statement does not depend on $n$ at all, but that this changes immediately for $k=3$. This is certainly an interesting topic for future research. Another question to be addressed is the investigation of the behavior of non-binary data. In \citep{pablo} the authors also considered this case briefly, but the exponential size of the respective alignments only allowed for exhaustive tree searches for up to $n=12$ taxa. A mathematical examination of this setting, particularly of the case of quaternary data like DNA or RNA, would be of high relevance also with regards to biological applications.

\section*{Acknowledgements} I wish to thank Mike Steel for insightful discussions on the topic as well as for suggesting the generalization of the example presented in Figure \ref{A3problem} to the case presented in Figures \ref{Anhalfproblem} and \ref{Anhalfproblem2}. Moreover, I want to thank Kristina Wicke and Lina Herbst for very helpful discussions on the topic of the present manuscript as well as concerning related questions. Last but not least, I wish to thank two anonymous reviewers for helpful comments on an earlier version of this manuscript.

\section*{Appendix} Here, we present a short proof for the fact that $T_1$ and $T_2$ from Figures  \ref{Anhalfproblem} and \ref{Anhalfproblem2} with $k>3$ share the same $A_k$ alignment (note that for the case $k=3$, we already presented alternative trees and their $A_3$ alignment in Figures \ref{A3problem} and \ref{problemal6taxak3}, so $k>3$ covers all remaining cases). As we did throughout the manuscript, we again assume without loss of generality that leaf 1 is in state $a$ for any character in $A_k$.

First consider $T_1$ on $n=2k$ leaves, where $k>3$. We now construct a set $B_k$ of characters with parsimony score $k$ on $T_1$, i.e. $B_k \subseteq A_k(T_1)$ as follows: $B_k$ shall consist of all binary characters which assign state $a$ to leaf 1, state $b$ to leaf $2k$, and all other cherries $[i,i+1]$ for $i =2, \ldots , 2k-2$ shall be such that the states assigned to the elements of the cherries are different. This leads to $|B_k|=2^{k-1}$, because the cherry $[1,2k]$ is identical for all characters in $B_k$, but all other $k-1$ cherries have two choices: $i$ can be assigned $a$ and $i+1$ can be assigned $b$ or vice versa. Note that each character in $B_k$ clearly has parsimony score $k$ on $T_1$ as depicted in Figure \ref{Anhalfproblem}, because no matter which state is chosen for the node to which $i$ and $i+1$ are adjacent, each cherry always requires one change. So the parsimony score of each character in $B_k$ is at least equal to $k$ (because the first cherry $[1,2k]$ also contributes one such change). On the other hand, for a binary character on $n=2k$ taxa, the maximum number of changes required is known to be $n/2=2k/2=k$ (cf. for instance \citep[Lemma 3.13]{fischer_kelk}). So all characters in $B_k$ must have parsimony score precisely $k$. Thus, $B_k \subseteq A_k(T_1)$, and $|B_k|=2^{k-1}$. However, using Theorem \ref{thm:lengthAk}, we get: $$|A_k(T_1)|=\frac{1}{2} \cdot  \frac{2\cdot 2k-3k}{k} \cdot {2k-k-1 \choose k-1} \cdot 2^k = 2^{k-1}.$$ So in total,  $B_k \subseteq A_k(T_1)$ and $|B_k|=2^{k-1}= |A_k(T_1)|$, which shows that $B_k=A_k(T_1)$. So all characters in $A_k(T_1)$ can be constructed by assigning leaf 1 state $a$, leaf $2k$ state $b$ and all cherries $[i,i+1]$ for $i =2, \ldots , 2k-2$ two different states, respectively.
%\par\vspace{0.3cm} 
%\noindent
 \par\vspace{0.5cm} 
 Next we consider the characters of $B_k$ on $T_2$ as depicted in Figure \ref{Anhalfproblem2}. Take any such character and highlight the path from leaf 1 to leaf $2k$, as well as the paths from $i$ to $i+1$ for each $i=2, \ldots, 2k-2$. Clearly, these $k$ paths are all edge-disjoint and they all connect a leaf in state $a$ with a leaf in state $b$. So the maximum number of such paths that can be found in such a character must be at least $k$. However, it is a well-known consequence of Menger's theorem that the maximum number of such paths equals the parsimony score of the given character on the tree under consideration (cf. for instance \citep[Lemma 5.1.7 and Corollary 5.1.8]{Semple2003}). So the parsimony score of any character of $B_k$ on $T_2$ must be at least $k$. However, as above, for a binary character on $n=2k$ taxa, the maximum number of changes required is known to be $n/2=2k/2=k$ \citep{fischer_kelk}. So, again, the parsimony score of any character in $B_k$ on $T_2$ must be exactly $k$, and thus $B_k \subseteq A_k(T_2)$ (and we have already seen that $|B_k|=2^{k-1}$). It remains to apply Theorem \ref{thm:lengthAk} again to see that $|A_k(T_2)|=2^{k-1}$, which immediately leads to $B_k=A_k(T_2)$. 

%\par\vspace{0.3cm} \noindent 
So altogether we have $A_k(T_1)=B_k=A_k(T_2)$, which completes the proof.

\bibliographystyle{model1-num-names}\biboptions{authoryear}
\bibliography{References}   % name your BibTeX data base

\end{document}